%% file: main.tex
\documentclass[11pt,letterpaper]{article}

\usepackage[bookmarks,colorlinks,breaklinks]{hyperref}
\hypersetup{urlcolor=blue, colorlinks=true, citecolor=green!50!black, linkcolor=blue}
\usepackage[letterpaper, left=1in, right=1in, top=0.9in, bottom=0.9in]{geometry}
\usepackage[utf8]{inputenc}
\usepackage[american]{babel}
\usepackage[normalem]{ulem}
\usepackage{amsmath, amssymb, cases, amsthm}
\usepackage{thmtools}
\usepackage[shortlabels]{enumitem}
\usepackage{mdframed}
\usepackage{bbm}
\usepackage{bm}
\usepackage{microtype}
\usepackage{xcolor}
\usepackage{makecell}
\usepackage{mathtools}
\usepackage{float}
\usepackage{modletters}
\usepackage{comment}
\usepackage{multirow}
\usepackage[capitalize,noabbrev]{cleveref}
\usepackage{graphics}

\declaretheorem[numberwithin=section,refname={Theorem,Theorems},Refname={Theorem,Theorems}]{theorem}

\declaretheorem[numberlike=theorem]{lemma}

\declaretheorem[numberlike=theorem,style=definition]{definition}
\declaretheorem[numberlike=theorem]{claim}
\declaretheorem[numberlike=theorem,style=remark]{remark}

\declaretheorem[numberlike=theorem, refname={Observation,Observations},Refname={Observation,Observations},name={Observation}]{observation}

\def\final{0}  % set this to 1 to get a comment-free version
\def\iflong{\iffalse}
\ifnum\final=0  %namely if we allow comments in the output
\newcommand{\yonggang}[1]{{\color{blue}[{\tiny Yonggang: \bf #1}]\marginpar{*}}}
\newcommand{\danupon}[1]{{\color{red}[{\tiny Danupon: \bf #1}]\marginpar{\color{red}*}}}
\newcommand{\sagnik}[1]{{\color{green!50!black}[{\tiny Sagnik: \bf #1}]\marginpar{\color{green!50!black}*}}}
\newcommand{\todo}[1]{{\color{red}[{\tiny TODO: \bf #1}]\marginpar{\color{red}*}}}
\newcommand{\yuval}[1]{{\bf \color{red!50!black} YUVAL: #1}}
\newcommand{\jan}[1]{{\bf \color{green!50!black} Jan: #1}}
\newcommand{\blikstad}[1]{\textup{\color{magenta} [\textbf{Joakim}: #1]}}
\newcommand{\TODO}[1]{{\color{blue!50!black} [{\bf Todo:} #1]}}
\else % in this case [final=1] we don't want any comments to show
\newcommand{\yonggang}[1]{}
\newcommand{\danupon}[1]{}
\newcommand{\sagnik}[1]{}
\newcommand{\todo}[1]{}
\newcommand{\yuval}[1]{}
\newcommand{\jan}[1]{}
\newcommand{\blikstad}[1]{}
\newcommand{\TODO}[1]{}
\fi  %ok, we're done with defining comment macros
\usepackage[ruled,vlined,linesnumbered,noresetcount]{algorithm2e}

\newcommand{\polylog}{\mathrm{polylog}}

\newcommand{\set}[2][ ]{\{#2 \ifthenelse{\equal{#1}{ }}{ }{~|~#1}\}}
\newcommand{\fvedge}{\mathsf{FindViolatingEdge}}
\newcommand{\XORQ}{\mathsf{XOR}}
\newcommand{\ANDQ}{\mathsf{AND}}

\newcommand{\R}{\mathbb{R}}

\newcommand{\BPM}{\mathsf{BPM}}
\newcommand{\BMM}{\mathsf{BMM}}
\newcommand{\UBPM}{\mathsf{UBPM}}

\newcommand{\MM}{\mathcal{P}}
\newcommand{\ORQ}{\mathsf{OR}}
\newcommand{\ISQ}{\mathsf{IS}}
\newcommand{\vol}{\mathrm{vol}}
\newcommand{\volumelb}{\left(\tfrac{1}{20n}\right)^{2n}}

\Crefname{algocf}{Algorithm}{Algorithms}

\title{Nearly Optimal Communication and Query Complexity of\\ Bipartite Matching}

\author{Joakim Blikstad\thanks{KTH Royal Institute of Technology, Sweden, \texttt{blikstad@kth.se}} \and Jan van den Brand\thanks{UC Berkeley \& Simons Institute, USA, \texttt{vdbrand@berkeley.edu}} \and Yuval Efron\thanks{Columbia University, USA, \texttt{ye2210@columbia.edu}} \and Sagnik Mukhopadhyay\thanks{University of Sheffield, UK, \texttt{s.mukhopadhyay@sheffield.ac.uk}} \and Danupon Nanongkai\thanks{University of Copenhagen, MPI-INF, and KTH, \texttt{danupon@gmail.com}}}
\date{}

\begin{document}
	
	\begin{titlepage}
		\maketitle \pagenumbering{roman}
		
		\input{abstract}

		\setcounter{tocdepth}{3}
		\newpage
		\tableofcontents
		\newpage
	\end{titlepage}
	
	\newpage
	\pagenumbering{arabic}

\section{Introduction} \label{sec:intro}

In the  \textit{maximum-cardinality bipartite matching} problem ($\BMM$), we are given a bipartite graph $G=(L\cup  R, E)$ with $n$ vertices on each side and $m$ edges. The goal is to find a matching of maximum size in $G$. 
This problem, along with its special case of \textit{bipartite perfect matching} ($\BPM$), 
are central problems in graph theory, economics, and computer science. They have been studied 
in various computational models such as the sequential, two-party communication, query, and streaming settings [See \cite{ford1956maximal,HopcroftK73,Lovasz79,KarpUW85,KarpVV90,MuchaS04,Zhang04,IvanyosKLSW12, Madry13,Madry16,GuruswamiO16,GoldwasserG17,BernsteinHR19,DobzinskiNO19,AnariV20,AhmadiK20,JinST20,BrandLNPSSSW20,Nisan21,AssadiR20,ChenKPS0Y21,FennerGT21,AssadiB21,AssadiLT21,ForsterGLPSY21,RoghaniSW22,ChenKLPGS22} and many more]. 
In this paper, we present simple algorithms and lower bound arguments that settle (up to polylog factors) the complexities of $\BMM$ and its generalizations (e.g.~max-cost matching and transshipment) in at least five models of computation. Our results answer open problems that have been raised repeatedly since at least three decades ago (e.g.~\cite{HajnalMT88, Zhang04, IvanyosKLSW12, DobzinskiNO19, Nisan21,BeniaminiUBPM}); see Table~\ref{table:bounds} for a summary of our results.

\paragraph{Communication complexity.} To be concrete, we start with the two-party communication model, where edges of the input graph $G$ are partitioned between two players Alice and Bob. The goal is for Alice and Bob to compute the value of the $\BMM$ or to decide if a $\BPM$ exists in $G$ by communicating as frugally as possible.
Many fundamental graph problems have been studied in this model since the 80s (e.g.~\cite{PapadimitriouS82,BabaiFS86,HajnalMT88,DurisP89}). For $\BMM$ and $\BPM$, their communication complexities have 
been extensively studied from several angles and perspectives, including exact solution protocols \cite{BabaiFS86, HajnalMT88, IvanyosKLSW12, DobzinskiNO19}, round restricted protocols \cite{FeigenbaumKMSZ05,GoelKK12,GuruswamiO16,DBLP:conf/soda/AssadiKL17,AssadiB19,AssadiR20}, multiparty protocols \cite{GuruswamiO16,HuangRVZ15,AssadiKLY16,Kapralov21,KapralovMT21,HuangRVZ20}, approximate solution protocols \cite{Kapralov21, KapralovKS14, KapralovKS14, KapralovMNT20,AssadiB21}, matrix rank and polynomial representation \cite{BeniaminiDegree, BeniaminiN21}, and economics and combinatorial auctions \cite{Roth82, Tennenholtz02, Bertsekas09}.
In particular, Hajnal, Maass, and Tur\'an \cite{HajnalMT88} showed a lower bound of $\Omega(n\log n)$ for deterministic protocols\footnote{\label{fn:cc-lower-bound-det}\cite{HajnalMT88} did, in fact, show this lower bound for $st$-connectivity, which, together with folklore reductions, imply the same bound for $\BPM$.}. For randomized and quantum protocols, the lower bounds are $\Omega(n)$ \cite{BabaiFS86,IvanyosKLSW12, Razborov92}\footnote{\label{fn:cc-lower-bound} The $\Omega(n)$ lower bound follows by a simple reduction from set-disjointness. \cite{HuangRVZ20} has shown a $\Omega(\alpha^2nk)$ lower bound for $k$-party point-to-point communication model for $\alpha$-approximation of $\BMM$. \cite{IvanyosKLSW12} shows a $\Omega(n)$ quantum communication lower bound by a reduction from inner-product in $\mathbb{F}_2$.}.
For an upper bound, Ivanyos, Klauck, Lee, Santha, and de~Wolf~\cite{IvanyosKLSW12} implemented the Hopcroft-Karp algorithm \cite{HopcroftK73} to get an $O(n^{3/2}\log n)$-bit deterministic protocol (see also \cite{DobzinskiNO19,Nisan21}).

Closing the large gap between existing upper and lower bounds has been mentioned as an open problem in, e.g., \cite{HajnalMT88, IvanyosKLSW12, DobzinskiNO19,Nisan21}.
Beniamini and Nisan \cite{BeniaminiN21} recently showed that the rank of the communication matrix is $2^{O(n \log n)}$, suggesting that a better upper bound might exist. On the other hand, $\Omega(n^2)$ lower bounds for $o(\sqrt{\log n})$-round communication may suggest that an $\Omega(n^{1+\Omega(1)})$ communication lower bound may exist \cite{FeigenbaumKMSZ05,GoelKK12,AssadiR20, ChenKPS0Y21}. 
In this paper, we resolve this open problem with an $O(n\log^2 n)$ upper bound:

\begin{theorem} \label{them:intro-cc}
The deterministic two-party communication complexity of $\BMM$ is $O(n \log^2 n)$.
\end{theorem}

Note that our protocol can find the actual $\BMM$ (Alice and Bob know edges in the $\BMM$ in the end) and not just its value. 
We can in fact solve a more general problem of min-cost bipartite perfect $b$-matching which implies upper bounds for a large class of problems due to existing reductions (see \cite{BrandLNPSSSW20}). 

\begin{restatable}{theorem}{weighted}
\label{thm:intro-other-problems}
Given that all the weights/costs/capacities are integers polynomially large in $n$, we can solve the following problems in the two-party edge-partition communication setting, using $O(n\log^2 n)$ bits of communication:   Maximum-cost bipartite perfect $b$-matching, Maximum-cost bipartite $b$-matching,  Vertex-capacitated minimum-cost $(s,t)$-flow,  Transshipment (a.k.a.~uncapacitated minimum-cost flow), Negative-weight single source shortest path,
Minimum mean cycle and Deterministic Markov Decision Process (MDP).
\end{restatable}

\begin{table}[]
\centering
\renewcommand{\arraystretch}{1.2}
\begin{tabular}{| l | c | c | c |}
\hline
\textbf{Models} & \multicolumn{2}{c|}{\textbf{Previous papers}} & \textbf{This paper} \\
\hline
& Lower bounds & Upper bounds &  \\
\hline
 &&&\\[-1.3em]
 Two-party communication & \makecell{$\Omega(n)$ Rand,\\$\Omega(n \log n)$ Det, \\ Footnote \ref{fn:cc-lower-bound-det} and \ref{fn:cc-lower-bound}} & \makecell{$\tilde O(n^{1.5})$ \\ \cite{DobzinskiNO19, IvanyosKLSW12}} & \makecell{$O(n \log^2 n)$, Det \\ \Cref{them:intro-cc}}\\ 
 &&&\\[-1.3em]\hline
 &&&\\[-1.3em]
   Quantum edge query & \makecell{$\Omega(n^{1.5})$ \\ \cite{Zhang04, BeniaminiDegree}} & \makecell{$O(n^{1.75})$ \\ \cite{LinL15}} & \makecell{$\tilde O(n^{1.5})$\\ \Cref{thm:intro-query}}\\ \hline
 &&&\\[-1.3em]
$\ORQ$-query & \makecell{$\Omega(n)$ Rand,\\$\Omega(n \log n)$ Det,\\ \cite{BeniaminiN21}} & \makecell{$\tilde O(n^{1.5})$ Det,\\ \cite{Nisan21}} &  \makecell{$O(n \log^2 n)$, Det \\ \Cref{thm:intro-query}}\\
  &&&\\[-1.3em]\hline
   &&&\\[-1.3em]
 $\XORQ$-query & \makecell{$\Omega(n)$ Rand \\$\Omega(n^2)$ Det \cite{BeniaminiN21}} & \makecell{$\tilde O(n^{1.5})$ Rand \\ \Cref{lem:simul-xor-or} and \cite{Nisan21}} & \makecell{$O(n \log^2 n)$, Rand \\ \Cref{thm:intro-query}}\\
&&&\\[-1.3em] \hline
 &&&\\[-1.3em]
  $\ANDQ$-query & \makecell{$\Omega(n)$ Rand,\\$\Omega(n^2)$ Det \cite{BeniaminiN21}} & \makecell{$O(n^2)$ \\ Trivial} & \makecell{$\Omega(n^2)$, Rand \\  \Cref{thm:intro-query}}\\ 
\hline
\end{tabular}
\caption{The communication and query complexity bounds for $\BMM$ and $\BPM$. All upper bounds are stated for $\BMM$ and all lower bounds are stated for $\BPM$.}
\label{table:bounds}
\end{table}

\paragraph{Query complexity.} 
Besides the communication complexity, we also settle the query complexity of $\BMM$ and $\BPM$ for several variants of the edge query model.  
In the standard edge query model, the querier can ask whether an edge in the input graph $G$ is present or not. The goal is to solve the graph problem by making as few queries as possible. This query model, in both deterministic and randomized settings, has been studied for almost half a century \cite{Rosenberg73, RivestV75, RivestV76, KahnSS84, Hajnal91, ChakrabartiK07} for various graph problems. For $\BPM$, Yao \cite{Yao88} showed that $n^2$ edge queries are necessary in the deterministic setting and D\"urr, Heiligman,  H{\o}yer, and Mhalla \cite{DurrHHM06} showed an $\Omega(n^2)$ lower bounds for the randomized setting\footnote{\cite{Yao88} showed a stronger result: Any non-trivial monotone graph property needs $n^2$ queries. \cite{DurrHHM06} mentioned a $\Omega(n^2)$ randomized query complexity for \textsc{Connectivity}. A similar construction (which is essentially a reduction from the query complexity of $\ORQ_{n^2}$) shows an $\Omega(n^2)$ lower bound for $(s,t)$-\textsc{Reachability} which reduces to $\BPM$.} thereby completely characterizing (up to constant factors) classical edge query complexity for $\BPM$.

However, for several variants of the classical edge query complexity, there are known gaps between the best known upper and lower bounds for $\BPM$. For example, in the case of \textit{quantum} edge query protocols, Zhang \cite{Zhang04} showed a lower bound of $\Omega(n^{1.5})$ by using Ambainis' adversary method \cite{Ambainis02} (see \cite{BeniaminiDegree} for an alternative proof via approximate degree). 
The best upper bound is, however, at $O(n^{1.75})$ as shown by \cite{LinL15}. This upper bound is obtained by simulating the Hopcroft-Karp algorithm using \textit{bomb queries} and relating it to the quantum edge queries.

Another well-studied variant of the classical query protocols is the $\XORQ$-query protocols (otherwise known as the \textit{parity decision trees}) where the querier is allowed to ask the following question about the input graph $G=(V, E)$: Given a set $S$ of potential edges of $G$, is $|S\cap E|$ odd or even?
Similarly, $\ANDQ$-queries and $\ORQ$-queries ask if $S \subseteq E$ or not and if $|S\cap E|\geq 1$ or not, respectively. 
Such query models have proven to be extremely important in the study of $\XORQ$-functions, the log-rank conjecture and lifting theorems \cite{KushilevitzM93, MontanaroO10, ChattopadhyayKL18, HatamiHL18, MandeS20}. As usual, these query models can be studied in deterministic, randomized and quantum models as well. For graph problems, these query models have recently started to receive increasing attention \cite{BeniaminiN21, BeniaminiUBPM, AssadiCK21}. For $\ANDQ$-query or $\XORQ$-query complexity, a recent result of \cite{BeniaminiN21} showed that $\Omega(n^2)$ queries are necessary to compute $\BPM$ deterministically\footnote{For $\XORQ$-queries, \cite{BeniaminiN21} showed that $\BPM$ is \textit{evasive}, i.e., requires $n^2$ queries.}. For $\ORQ$-query, \cite{BeniaminiN21} also showed a deterministic lower bound of $\Omega(n \log n)$. The upper bound of $\tO(n^{1.5})$ for $\ORQ$-queries (and, thereby, \textit{randomized} $\XORQ$-queries, see \Cref{lem:simul-xor-or}) can be achieved by simulating the Hopcroft-Karp algorithm \cite{Nisan21}.

From the above results, it remained open to close the polynomial gaps for quantum and $\ORQ$-queries (as mentioned in \cite{Nisan21, BeniaminiDegree}) and whether randomization helps for $\XORQ$-queries and $\ANDQ$-queries. In this paper, we answer these questions: We provide upper bounds that are tight up to polylogarithmic factors for quantum and $\ORQ$-queries. Our upper bound result also shows that randomization helps for $\XORQ$-queries. In contrast, for $\ANDQ$-queries we can show that an $\Omega(n^2)$ lower bound holds even for randomized algorithms. Our results are summarized below and in \Cref{table:bounds}. Note that our lower bound argument also gives simplified proofs of the lower bounds for $\XORQ$-queries and $\ORQ$-queries.

\begin{theorem} \label{thm:intro-query}
The following query bounds hold for $\BMM$: \begin{itemize}[noitemsep]
    \item The quantum edge query complexity  is $O(n^{1.5}\log ^2 n)$,
    \item The deterministic $\ORQ$-query complexity is $O(n \log^2 n)$,
    \item The randomized $\XORQ$-query complexity is $O(n \log^2 n)$,
\end{itemize}
Moreover, the randomized $\ANDQ$-query complexity of $\BPM$ is $\Omega(n^2)$.
\end{theorem}

Finally, our results also extend to the unique bipartite perfect matching problem ($\UBPM$), which has been studied in, e.g., the sequential and parallel settings \cite{KozenVV85, GabowKT01, HoangMT06, BeniaminiUBPM}. Beniamini~\cite{BeniaminiUBPM} recently show $\UBPM$ lower bounds similar to those for $\BMM$ and $\BPM$, i.e. $\Omega(n \log n)$ communication complexity, $\tilde\Omega(n^{1.5})$ quantum edge query (under a believable conjecture\footnote{\cite{BeniaminiUBPM} conjectured that the approximate degree of  $\UBPM$ is $\Omega(n^{1.5})$ (see Conjecture 1) which would imply a similar lower bound for quantum edge query complexity.}), $\Omega(n \log n)$ $\ORQ$-queries, $\Omega(n^2)$ $\XORQ$-queries, and $\Omega(n^2)$ $\ANDQ$-queries. We complement these lower bounds with tight upper bounds, i.e.~$O(n \log^2 n)$ deterministic communication protocol, $O(n^{1.5}\log^2 n)$ quantum edge query algorithm, $O(n \log^2 n)$ deterministic $\ORQ$-query and randomized $\XORQ$-query algorithms, and $\Omega(n^2)$ randomized $\ANDQ$-query lower bound. Our upper bounds answer an open problem by Beniamini \cite{BeniaminiUBPM}. 
 
\medskip\noindent
\textit{Update:} After our paper was accepted in FOCS 2022, we observed that our technique also leads to a $O(n \log^2 n)$ deterministic protocol in the well-studied \textit{Indenpendent set} ($\ISQ$) query model \cite{BeameHRRS18,AuzaL21, RashtchianWZ20, AlonA05, AlonBKRS04, AbasiN19}. In this model, a query consists of two disjoint subsets of vertices $X$ and $Y$, and the answer to the query is 1 iff there is an edge between $X$ and $Y$ (i.e., $E \cap (X \times Y) \neq \emptyset$).

\paragraph{Organization.} In \Cref{sec:tech-overview}, we provide a brief technical overview of our upper and lower bounds. In \Cref{sec:open-problems}, we list a few open problems that naturally arise from our work. \Cref{sec:BPM} details our various upper bounds, starting with $\ORQ$-query protocols. In \Cref{ssec:applications}, we show the applications of the $\ORQ$-query algorithm, namely two party communication complexity (\Cref{sec:appl-cc}), randomized $\XORQ$-query (\Cref{sec:appl-xor}), Independent set query (\Cref{sec:appl-is}), $\ORQ_k$-query (\Cref{sec:appl-ork}) and quantum edge query (\Cref{sec:quantum}). We then list different variants of the bipartite matching problem (\Cref{sec:weighted}) that our technique can solve as well. Finally, in \Cref{sec:communication-lower-bounds}, we provide lower bounds for solving $\BPM$ in $\ORQ$-, $\ANDQ$- and $\XORQ$-query settings.

\subsection{Technical Overview} \label{sec:tech-overview}

\paragraph{Upper bounds.}
Our algorithms follow an existing continuous optimization method. There are many such methods and the question is: {\em what is the right method?}
An intuitive idea would be to implement some fast sequential algorithms for $\BMM$ and related problems (e.g. \cite{DaitchS08,Madry13,LeeS14,Madry16,cmsv17,CohenLS19,b20,LiuS20,AxiotisMV20,BrandLSS20,BrandLNPSSSW20,BrandLLSS0W21-maxflow,ChenKLPGS22}), which are based on {\em central path methods}.
It is not clear, however, how to implement central path methods efficiently in query or communication settings. They require polynomially many iterations (e.g.~$\Omega(\sqrt{n})$), each of which needs a large communication and query complexity (e.g.~$\Omega(n)$ per iteration).
Another option is to use one of the {\em cutting planes methods} (e.g.~the Ellipsoid method). These methods are a framework for solving general convex optimization problems and thus are rather slow for $\BMM$ in the sequential setting (e.g.~$\tO(mn)$ time \cite{LeeSW15}) compared to more specialized alternatives based on central path methods.
However, it turns out that cutting planes methods are the right framework for the communication and query settings!  
In particular, we can implement a cutting planes method with a low number of iterations, such as the {\em center-of-gravity}  (CG)  and \emph{volumetric center} (VC) methods \cite{Levin65,newman1965location,Vaidya89}, 
on the {\em dual linear program}, i.e.~the minimum vertex cover linear program\footnote{%
%We note here that a recent result \cite{VempalaWW20} also uses the cutting plane method to solve general linear program in the two-party communication setting. However, applying their technique as a black-box will not yield the best possible complexity for $\BMM$---our novelty is in using this method to solve the \textit{dual} LP in order to reduce the number of iterations which, in turn, produces the efficient protocol.
We thank an anonymous FOCS'22 reviewer for pointing out the result in \cite{VempalaWW20} that uses the cutting plane method to solve general linear program in the multiparty model of communication. Our result is independent and follows the same general cutting plane framework but we exploit that for our specific linear program,
(i) cutting planes have short description
and
(ii) we have a better bound on the number of iterations.
%applying their technique as a black-box will not yield our result.
%the best
%possible complexity for BMM -- we crucially use the facts that for our specific linear program,
%(i) cutting planes have short description
%we work with an LP whose constraints have small representation, 
%and
%(ii) we have a better bound on the number of iterations.
%the ratio of start- to end-volume of the polytope is small.
%in the dual LP, cutting planes have short description.
}.
(We cannot use the Ellipsoid method due to its high number of iterations.)
The CG and VC methods are not useful for solving $\BMM$ in the sequential setting due to their high running time (the CG method even requires exponential time); however, this high running time is hidden in the internal computation and thus does not affect the communication/query complexities.

Using the cutting planes methods above, our algorithm is simply the following: 
We start with an assignment $p:V\rightarrow \mathbb{R}^+$ on the vertices that is supposed to be a fractional vertex cover of value $F$, i.e.~for every edge $(u,v)$, $p(u)+p(v)\geq 1$ and $\sum_{v\in V(G)}p(v)\leq F$.
In each iteration, we need to find a {\em violated constraint}, i.e.~an edge $(u,v)$ such that $p(u)+p(v)< 1$, or the value constraint if $\sum_{v\in V(G)}p(v) > F$. 
This violated constraint then allows us to compute a new assignment $p:V\rightarrow \mathbb{R}^+$ (which is the center of gravity of some polytope) to be used in the next iteration. 
It can be shown that this process needs to repeat only for $\tO(n)$ times to construct a fractional vertex cover of value at least $F$, or conclude no such cover exists. 

This simple algorithm leads to efficient algorithms in many settings. For example, in the two-party communication setting, Alice and Bob only need to communicate one violated constraint in each iteration while they can compute the new assignment  $p:V\rightarrow \mathbb{R}^+$ without any additional communication ($p:V\rightarrow \mathbb{R}^+$ depends only on the discovered violated constraints and not on the input graph). It is also not hard to implement this method in other settings.
We note that in this paper we use the CG method for simplicity. This method leads to exponential internal computation. This can be made polynomial by using the VC method \cite{Vaidya89} instead.

\paragraph{Lower bounds.} 
For lower bounds, our goal is to prove a lower bound for $\BPM$ (which also implies a lower bound for $\BMM$). 
Let us start with our randomized $\ANDQ$-query lower bound of $\Omega(n^2)$. A typical approach to show this is proving an $\Omega(n^2)$ communication complexity lower bound in the setting defined earlier; however, we have already shown in \Cref{them:intro-cc} that this is not possible. 
\cite{BeniaminiN21} sidestepped this obstacle by considering the real polynomial associated with $\BPM$.
Known connections between the monomial complexity of this polynomial and $\ANDQ$-query complexity yield corresponding tight $\Omega(n^2)$ \textit{deterministic} $\ANDQ$-query complexity for $\BPM$.

It turns out that we can prove a {\em randomized}
$\ANDQ$-query lower bound (and simplifying the lower bounds proofs of \cite{BeniaminiN21}) by revisiting the two-party communication lower bounds, but with a slightly different definition. Our main observation here is that $\ANDQ$-queries can be simulated cheaply by the following variant of the two-party communication model: Alice gets edge set $E_A \subseteq E$, Bob gets edge set $E_B \subseteq E$, and they solve $\BPM$ (or any other graph function) in the graph $G_\cap=(V, E_A \cap E_B)$. Our $\ANDQ$-query lower bound now follows from a reduction from the set disjointness problem.

Similarly, Beniamini and Nisan \cite{BeniaminiN21} use real polynomial techniques to prove deterministic lower bounds for $\XORQ$-queries and $\ORQ$-queries. We provide simple alternative proofs via the communication complexity of $\BMM$ in the symmetric difference and union graphs $G=(V, E_A\oplus E_B)$ and $G=(V, E_A\cup E_B)$; such lower bounds can be proved via a reduction from the equality and $st$-reachability problems. Finally note that even though we simplify the query lower bounds proofs, \cite{BeniaminiN21, BeniaminiDegree} showed something stronger, i.e., a complete characterization of the unique multilinear polynomial over reals representing $\BPM$ which may have other interesting consequences beyond query complexity.

\subsection{Open problems}
\label{sec:open-problems}

The communication complexity of $\BMM$ and $\BPM$ has been a bottleneck for many tasks. The fact that it can be solved by a simple cutting planes method might be the gateway to solving many other problems. Below we list some of these problems.

\begin{enumerate}
    \item \textbf{Demand query complexity of $\BMM$.}
    The demand query setting is equivalent to when we can issue an $\ORQ$-query only on the edges incident on a single left vertex (or, equivalently, an $\ISQ$-query where set $X \subseteq L$ is singleton). Minimizing the number of demand queries used to solve $\BMM$ and $\BPM$ is motivated by economic questions  \cite{Nisan21, BeniaminiDegree}. Like in many settings we consider, the best demand query upper and lower bounds for $\BMM$ and $\BPM$ are $O(n^{1.5})$ and $\Omega(n)$ respectively. Closing this gap remains open. Because of our efficient $\ISQ$-query protocol, we believe that a possible direction is to extend our approach to get a better upper bound for demand query. For a better lower bound, our results suggest that one might need a technique specialized for the demand query lower bound: the two known approaches for proving a demand query lower bounds are via quantum and $\ORQ$-queries (see, e.g., Figure 7 in \cite{BeniaminiDegree}) and our quantum and $\ORQ$-query upper bounds show that these approaches cannot be used.

    \item {\bf Bounded communication rounds and streaming passes.}  Most graph problems, including $\BMM$ and $\BPM$, admit an $\Omega(n^2)$ communication lower bound when only Alice can send a message (i.e. the one-way communication setting) \cite{FeigenbaumKMSZ05}. If Bob gets to speak back once (the 2-round setting), some problems become much easier (e.g. the communication complexity of global edge connectivity reduces from $\Omega(n^2)$ to $\tilde O(n)$) \cite{AssadiD21}. Unfortunately, such an efficient protocol for $\BPM$ does not exist even when we allow $o(\sqrt{\log n})$ rounds \cite{ChenKPS0Y21,AssadiR20}. More generally, $r$-round protocols are known to require $n^{1+\Omega(1/r)}$ communication \cite{AssadiR20,GuruswamiO16}.
    An important question is to get tight $r$-round communication bounds for $\BMM$ and $\BPM$. 
    Our algorithm provides an $\tO(n)$ communication bound for the extreme case where $r=n$. One possible extension is to study {\em bounded-iteration} cutting planes methods. For example, can we reduce the number of iterations if in each iteration we can identify more violating constraints?
    It will be exciting if a $\polylog(n)$-round $\tilde O(n)$-communication protocol exists. 
    It will be even more exciting if this can be extended to a {\em $\polylog(n)$-passes streaming algorithm} (breaking \cite{LiuSZ2020-stream,AssadiJJST22} and matching \cite{GuruswamiO16}). %For the more general problem of max-flow, lower bounds in the streaming model are known under certain conditions \cite{AssadiCK19}.

     \item \textbf{Distributed Matching.} The distributed CONGEST model is an important model to study fundamental graph problems (e.g. minimum spanning tree, shortest paths, and minimum cut) on distributed networks (e.g. \cite{GallagerHS83,KuttenP98,Nanongkai14,GhaffariL18,ForsterN18,Elkin20a,BernsteinN19,AgarwalR20,GhaffariKKLP18,HenzingerKN21,ChechikM20,Ghaffari0T20,NanongkaiS14,DoryEMN21}). Compared to other graph problems, computing $\BMM$ and $\BPM$ exactly in CONGEST is much less understood in this model. This is despite the studies of their variants since the 80s~\cite{Luby86,IsraelI86,Gall16,AhmadiKO18,AhmadiK20}. The best lower bound for this problem is  $\Tilde{\Omega}(\sqrt{n}+D)$\cite{AhmadiKO18,SarmaHKKNPPW12} (see also \cite{HaeuplerWZ21}).
    The best upper bound is $O(n\log n)$ \cite{AhmadiKO18}. For sparse graphs, the upper bound can be improved to $\tO(m^{3/7}(\sqrt{n}D^{1/4}+D))$ via continuous optimization~\cite{ForsterGLPSY21}. (Better upper bounds via fast matrix multiplication also exist on the special case of {\em congested clique} \cite{Gall16}.)
    A major open problem is to close the gap between upper and lower bounds. 
    Our results may suggest a new approach for improving the known upper bounds for the problem.
   Past results seem to suggest that graph problems with $\tO(n)$ communication complexity usually admit an $\tO(\sqrt{n}+D)$ upper bound in CONGEST. (A recent example is the  $\tO(n)$ communication complexity protocol of mincut~\cite{MukhopadhyayN20} that was later extended to achieve an $\tO(\sqrt{n}+D)$ upper bound in CONGEST \cite{DoryEMN21}.) 
Proving that this is or is not the case for $\BMM$ and $\BPM$ will be an exciting result.

     \item \textbf{General Matching.}
The maximum matching problem on \emph{general} (i.e.\ not-necessarily-bipartite) graphs is less understood than that on bipartite graphs. Unlike $\BMM$, the linear programming formulations for general matching is rather unwieldy, making it difficult to apply the cutting planes method approach. Settling the communication and query complexity of general matching remain intriguing open problems.  On one hand, there might be a hope to show truly super-linear  (i.e., $\Omega(n^{1+\epsilon})$ for some constant $\epsilon > 0$) communication lower bounds in these models, thereby showing a gap between the bipartite and non-bipartite case. On the other hand, an $\tilde{O}(n)$ communication complexity upper bound for the general matching problem would hopefully shed some light on the interplay between matchings on bipartite versus general graphs.

    \item \textbf{Maxflow/mincut and Related Problems.} 
    Max $(s,t)$-flow, equivalently min $(s,t)$-cut, is a powerful tool that can be used to solve $\BMM$, $\BPM$, and many other fundamental graph problems. Efficiently solving this problem  could only be a dream in the past in many computational models since even its special case of matching could not be solved efficiently. Our results serve as a step toward this goal. 
    Particularly interesting goals are solving $(s,t)$-max-flow/min-cut in the communication\footnote{Here we expect Alice and Bob to know the flow values in their respective sets of edges. The decision version of this problem where we ask if the total flow is at least a threshold $k$ is also interesting.}, distributed, cut query, and streaming settings (Bounded round communication lower bounds in multiparty communication setting for $(s,t)$-max-flow/min-cut have been studied in \cite{AssadiCK19}.).
    Also, there are problems that were recently shown to be solvable in max-flow time in the sequential setting such as Gomory-Hu tree, vertex connectivity, Steiner cut, hypergraph global min-cut, and edge connectivity augmentation \cite{ChuzhoyGLNPS20, LiP20, LiP21, LiNPSY21, ChekuriQ21a, MukhopadhyayN21}. Can these problems be solved as efficiently as max-flow in other settings, e.g.~the communication, distributed, and streaming settings?

\end{enumerate}

Other problems include (i) showing $\Omega(n\log n)$ \emph{randomized} communication lower bound for connectivity or even just for $\BMM$ and min-cost flow, (ii) closing the $\log n$ factor gap between $\ORQ$-query upper and lower bounds, and (iii) settling the quantum $\ORQ$-query and $\ANDQ$-query complexity of $\BMM$ and $\BPM$.%, and (iii) settling various query and communication complexity measures of matching on general (non-bipartite) graphs.

\input{or_bmm_section}

\input{max-cost-b-matching}

\input{communication_lower_bounds}

\section*{Acknowledgement}
Jan van den Brand is partially funded by ONR BRC grant N00014-18-1-2562 and by the Simons Institute for the Theory of Computing through a Simons-Berkeley Postdoctoral Fellowship.

This project has received funding from the European Research
        Council (ERC) under the European Union's Horizon 2020 research
        and innovation programme under grant agreement No
        715672. Nanongkai was also partially supported by the Swedish
        Research Council (Reg. No. 2019-05622).
        
The authors would like to thank Gal Beniamini for fruitful discussions and anonymous reviewers of FOCS 2022 for their useful suggestions.

\bibliography{References}

\appendix

\input{appendix_maxcost}

\input{appendix_communication_lower_bound}

\end{document}

%% file: abstract.tex
\begin{abstract}
We settle the complexities of the maximum-cardinality bipartite matching problem (BMM) up to poly-logarithmic factors in five models of computation: the two-party communication, AND query, OR query, XOR query, and quantum edge query models. 
Our results answer open problems that have been raised repeatedly since at least three decades ago [Hajnal, Maass, and Turan STOC'88; Ivanyos, Klauck, Lee, Santha, and de Wolf FSTTCS'12;  Dobzinski, Nisan, and Oren STOC'14; Nisan SODA'21] and tighten the lower bounds shown by Beniamini and Nisan [STOC'21] and Zhang [ICALP'04]. 
We also settle the communication complexity of the generalizations of BMM, such as maximum-cost bipartite $b$-matching and transshipment; and the query complexity of unique bipartite perfect matching (answering an open question by Beniamini [2022]). 
Our algorithms and lower bounds follow from simple applications of known techniques such as cutting planes methods and set disjointness.
\end{abstract}

%% file: or_bmm_section.tex
\section{Bipartite Matching Upper Bounds}\label{sec:BPM}

Our goal in this section is to present a simple  $\ORQ$-query algorithm based on the cutting planes framework to find a maximum matching of a bipartite graph, i.e. to solve the $\BMM$ problem. From there we show how our $\ORQ$-query algorithm can be translated to several other information theoretical models of computation. Formally, the following is the main theorem of the section.

\begin{theorem}\label{thm:BPM_in_all_models}
Given $n$, there are algorithms solving $\BMM$ in the following models.
\begin{enumerate}
    \item Deterministic two-party edge-partition communication, with communication complexity $O(n\log^2 n)$.
    \item Deterministic $\ORQ$-query, with query complexity  $O(n\log^2 n)$.
    \item Randomized $\XORQ$-query, with query complexity $O(n\log^2 n)$.
    \item Quantum edge query, with query complexity $O(n^{1.5}\log^2 n)$.
\end{enumerate}
\end{theorem}
\paragraph{Overview.} 
We employ a standard cutting planes framework to determine if a bipartite graph has a vertex cover of a given size $F$ or not. We show that this cutting planes method can be implemented in $O(n\log n)$ iterations, where in each iteration we access the input graph a small number of times ($O(\log n)$) using $\ORQ$-queries to find an edge that corresponds to a violated constraint (i.e.\ a cutting plane), if one exists. 
Throughout this work, we use the following well known characterization of the existence of a matching of a certain size in a bipartite graph.

\begin{claim}[K\"onig's Theorem]\label{Claim:Konig_claim}
A bipartite graph $G$ has a minimum vertex cover of size $F$ if and only if it does not have a matching of size $F+1$.
\end{claim}

\paragraph{The vertex cover linear program.}
For a bipartite graph $G=(V,E)$ with $V = L\cup R$, $|L|=|R|=n$, the following linear program \eqref{eq:MM} over $x\in \bbR^{V}$ describes the fractional minimum vertex cover problem on $G$. Since $G$ is bipartite, the constraint matrix is totally unimodular, and hence $(\MM^G)$ is integral \cite[Section 5]{korte2011combinatorial}, i.e. there exists an integer optimal solution to $(\MM^G)$.

\begin{equation}
\begin{array}{ll@{}ll}
\text{minimize}  & \displaystyle\sum_{v\in V} x_{v} &\\
\text{subject to}&\displaystyle x_u+x_v\geq 1   && \forall (u,v)\in E\\
                            &0\le x_{v}\le 1  &&\forall v\in V
\end{array}
\tag{$\MM^G$}\label{eq:MM}
\end{equation}

\paragraph{Decision version.}
We first consider a \emph{decision version} of our problem, namely given an integer $F$ we want to determine if $G$ has a matching of size at least $F+1$.
Note that if we can solve this decision version, then we can also---by binary-searching over $F$---solve the \emph{optimization version} (i.e.\ finding the minimum size of a vertex cover / maximum size of a bipartite matching) with an overhead of $O(\log n)$. We start by focusing on solving the decision version (\cref{ssec:decision_alg}), and later (in \cref{ssec:optimization_alg}) we show how to, via a simple modification of the algorithm, actually solve the optimization version \emph{without} this extra $O(\log n)$ binary-search overhead.

By K\"onig's Theorem (\cref{Claim:Konig_claim}), determining whether $G$ has a matching of size at least $F+1$ is equivalent to determining whether $G$ (does not) have a vertex cover of size at most $F$.
This is equivalent to determining if $(\MM^G)$ has some feasible solution $x$ with $\sum x_v \leq F$. So we define another polytope \eqref{eq:MMF} as follows:

\begin{equation}
\begin{array}{ll@{}ll}
& \displaystyle\sum\limits_{v\in V} x_{v}  \leq F+\tfrac{1}{3}&\\
&\displaystyle x_u+x_v\geq 1   && \forall (u,v)\in E\\
                            &0\le x_{v}\le 1  &&\forall v\in V
\end{array}
\tag{$\MM_F^G$}\label{eq:MMF}
\end{equation}

Our decision algorithm either finds a feasible point for the above polytope, or it finds a witness of $\MM_F^G$ having no feasible points in the form of a set of edges that contains a matching of size $F+1$.

Note that we relax the constraint
$\sum x_v \leq F$ a bit to $\sum x_v \leq F+\frac{1}{3}$. This ensures that our polytope has a significantly large volume if it is non-empty (see \cref{claim:volume_lower_bound}). Thus our cutting planes methods can terminate and conclude that the polytope is empty whenever the volume is too small. This relaxation does not impact the correctness of our algorithm: since $(\MM^G)$ is integral, it has an integral optimal objective value, which means that if a feasible solution $x$ of $(\MM_F^G)$ exists, then there also exists a feasible solution $x'$ which achieves 
$\sum x'_v \leq F$.

\begin{lemma}\label{claim:volume_lower_bound}
For any bipartite graph $G=(V,E)$, if $F$ is an integer such that $(\MM_F^G)$ is non-empty, then $\vol(\MM_{F}^G)\geq \volumelb$.
\end{lemma}
\begin{proof}
Let $x$ be an integral solution for $(\MM_F^G)$, of value $F$. Indeed, if $(\MM_F^G)$ is feasible, then such an $x$ must exist due to the integrality of $(\MM^G)$. Let $I_{0}=\set{i\in [2n]\mid x_i = 0}$ and $I_{1}=\set{i\in [2n]\mid x_i = 1}$. We argue that the hypercube
$
[\tfrac{1}{20n}, \tfrac{1}{10n}] ^{I_0}
\times
[1-\tfrac{1}{20n}, 1]^{I_1}
$ is completely contained in $(\MM_F^G)$.
That is, if, for each $x_i$ with $x_i = 1$ we replace it with any value in $[1-\tfrac{1}{20n}, 1]$;
and for each $x_i$ with $x_i = 0$ we replace it with
with any value in $[\tfrac{1}{20n}, \tfrac{1}{10n}]$; the point remains feasible for $(\MM_F^G)$. We verify this below.

\begin{itemize}
\item The $0\le x_v \le 1$ constraints remain valid.
\item   Similarly, the $\sum x_v \leq F+\frac{1}{3}$ constraint remains valid, since we increase the value of $x_i$ by at most $\tfrac{1}{10n}$ for each $i$, and there are $2n$ vertices in total (so we increase $\sum x_v$ by at most $\tfrac{1}{5}$).
\item Lastly, the constraint $x_u + x_v \geq 1$ (for an edge $(u,v)\in E$) also remains valid, as either (i) both $x_u$ and $x_v$ were $1$ before, in which case we now have $x_u+x_v \geq 2-\frac{1}{10n}$; or (ii) exactly one of $x_u$ or $x_v$ was 1 before, in which case we increased the variable which was $0$ by at least $\tfrac{1}{20n}$ and decreased the variable which was $1$ by at most $\tfrac{1}{20n}$.
\end{itemize}

Thus we have argued that a hypercube of volume $\volumelb$ is contained in $(\MM_F^G)$.
\end{proof}

\subsection{$\ORQ$-query decision algorithm }\label{ssec:decision_alg} In this section we describe our cutting planes based  $\ORQ$-query algorithm for solving the feasibility problem on \eqref{eq:MMF}. We begin with a verbal overview of the algorithm, followed by pseudocode in \cref{alg:two}.
 The main lemma of this section is the following. 
 
\begin{lemma}\label{lemma:OR_query_result}
Given an integer $F$, there is a deterministic algorithm (\cref{alg:two}) using $O(n\log^2 n)$ $\ORQ$-queries which on an input bipartite graph $G=(V,E)$ either finds a feasible point in \eqref{eq:MMF}, or else a witness, in the form of a matching of size $F+1$, that \eqref{eq:MMF} is empty.
\end{lemma}

\paragraph{Center-of-gravity cutting planes method.} We are now ready to introduce the cutting planes framework \cite{Levin65,newman1965location}.
The idea is that we start with the polyhedra $P_0 = \{x\in [0,1]^{V} : \sum x_v \leq F+\tfrac{1}{3}\}$ (which contains $(\MM_F^G)$), and repeatedly find ``good'' constraints ``$x_u+x_v\geq 1$'' (corresponding to edges $(u,v)\in E$) to add which reduce the volume sufficiently fast.
Eventually, we either find a (fractional) feasible solution to $(\MM_F^G)$, or have determined that no such feasible point exist.

We work in iterations, each iteration $i$ is characterized by a polyhedron $P_i\supseteq (\MM_F^G)$. 
We compute the \emph{center-of-gravity} of $P_i$, denoted by $p_i = cg(P_i)\in P_i$, and defined to be $cg(P_i) = \left(\int_{P_i} z\, \mathrm{d}\!z\right)/\left(\int_{P_i} \mathrm{d}\!z\right)$. Note that we know $P_i$, so our algorithm can compute\footnote{%
Finding the center-of-gravity in an $n$-dimensional polyhedron is $\mathsf{NP}$-hard. 
However, all the considered models in \cref{thm:BPM_in_all_models} are query models,
and in particular are purely information-theoretical, and we can thus disregard computational concerns.
For ease of presentation, we work with center of gravity, but alternatively, one could use other variants of cutting plane using more computationally efficient notions of ``center'' such as volumetric centers \cite{Vaidya89}.}
$p_i = cg(P_i)$ without using any queries.

Either $p_i$ is feasible for $(\MM_F^G)$, in which case the cutting planes algorithm reports this and terminates. Otherwise there must exist some \emph{violated constraint} ``$x_u + x_v \geq 1$'' in $(\MM_F^G)$ but not in $P_i$ (i.e.\ $p_i$ does not satisfy this constraint, that is $p_i^u + p_i^v < 1$).
In this case, we want to find such a violated constraint, and let $P_{i+1} = P_i \cap \{x\in \R^{V} : x_u + x_v \geq 1\}$, after which we continue with the next iteration of the cutting planes method on $P_{i+1}$. We say that an edge $(u,v)\in E$ is a \emph{violating edge} for iteration $i$ if $p_i^u+p_i^v<1$. The process of finding a violating edge is the only part
of the algorithm which requires access to the input graph, and hence the only place where $\ORQ$-queries are being issued. Essentially, we need to implement a \emph{separation oracle} $\fvedge$, which we explain how to do with $\ORQ$-queries in \cref{clm:fvedge}. The full algorithm can be found in \cref{alg:two}.

\begin{algorithm}
\SetEndCharOfAlgoLine{}

\SetKwInput{KwData}{Input}
\SetKwInput{KwResult}{Output}

\caption{$\ORQ$-query algorithm for $\BMM$}\label{alg:two}
 \KwData{$\ORQ$-query access to $G=(L\cup R,E)$, vertex set $L\cup R$, feasibility parameter $F$}
 \KwResult{Whether $(\MM_F^G)$ is feasible}
$P_0 \gets \left\{x\in[0,1]^{2n}\ \Big\vert\ \sum\limits_{v\in V} x_v\leq  F+\tfrac{1}{3}\right\}$\;
$E' \gets \emptyset$\;
$i \gets 0$\;
\While{$vol(P_i) \geq \volumelb$}{
  $p_i \gets cg(P_i)$\;
  $(u,v)\gets \fvedge(E',p_i)$\;\label{step:crucial_step}
  \If{no edge was found}{
      \textbf{return} ``Feasible''
      \tcp*{$\textrm{$p_i$ is feasible for $(\MM_F^G)$}$}
  }
  $E' \gets E' \cup \set{(u,v)}$\;
  $P_{i+1}\gets P_{i}\cap \{x\in \R^{2n} \mid x_{u}+x_{v}\geq 1$\}\;
  $i \gets i+1$\;
}
\textbf{return} ``Infeasible''
\tcp*{$\textrm{$E'$ contains a matching of size $F+1$}$}
\end{algorithm}

\begin{claim}[$\ORQ$-implementation of $\fvedge$]
\label{clm:fvedge}
Using $O(\log n)$ $\ORQ$-queries we can find a violating edge or else determine that none exist.
\end{claim}
\begin{proof}
Given the center-of-gravity point $p_i$, we let $S=\set{(u,v)\in L\times R\mid p_i^u+p_i^v < 1}$ be the set of pairs of vertices $(u,v)$ which would be a violating edge if this pair was also an edge of the graph. Our task is thus to find some edge $e\in S\cap E$, or else determine that $S\cap E$ is empty. This can be done by a binary-search (with $\ORQ$-queries) over $S$.
\end{proof}

We now turn to prove several properties about our \cref{alg:two}.

\begin{observation}\label{obs:polytopes_contain}
Let $i$ be some iteration of the execution of \cref{alg:two}, then $P_i\supseteq \MM_F^G$.
\end{observation}

\begin{proof}
For every $i$, the set of constraints defining $P_i$ is, by the behaviour of the algorithm, a subset of the constraints defining $\MM_F^G$, thus the observation follows.
\end{proof}

\begin{lemma}\label{claim:number_of_iterations_bound} 
The algorithm terminates after $O(n \log n)$ iterations of the cutting planes method. 
\end{lemma}
\begin{proof}
We use the following well-known property of the center of gravity of a convex polytope.
\begin{lemma}[\cite{Grnbaum1960}]
\label{lem:volume-reduction}
For any convex polytope $P$ with center of gravity $c$ and any halfspace $H =\{x \mid \langle a, (x-c) \rangle \geq 0\}$ passing through $c$, it holds that:
\[
\frac 1 e \leq \frac{\vol(P \cap H)}{\vol(P)} \leq \left( 1 - \frac 1 e \right).
\]
\end{lemma}

This implies that, in our case, $\vol(P_{i+1})\leq (1-\tfrac{1}{e})\vol(P_i)$.
This means that in each iteration, we either find a feasible solution to $(\MM_F^G)$, or cut down the volume by a constant fraction as we have found a violating edge. Initially, $\vol(P_0) \le 1$, since it is contained in the unit-hypercube $[0,1]^{2n}$. By \cref{claim:volume_lower_bound} we can terminate when $P_i$ has volume less than $\volumelb$ and conclude that $(\MM_F^G)$ is empty in this case. This happens after at most $O\left(\log ((20n)^{2n})\right) = O(n \log n)$ iterations.
\end{proof}

\begin{lemma}\label{lemma:correctness_meta_algorithm}
Let $i_{max}$ denote the last iteration in the execution of the algorithm.
Then either $p_{i_{max}}\in \MM_F^G$ which serves as a witness that a vertex cover of size $F$ exists, or $\MM_F^G=\emptyset$ and the set $E'\subseteq E$ (constructed by the algorithm) contains a matching of size $F+1$.
\end{lemma}
\begin{proof}

In the case where we find a feasible point $p$ in $(\MM_F^G)$, this point is a fractional vertex cover of size at most $F+\frac{1}{3}$ for our graph (and hence a non-constructive witness that there exists an (integral) vertex cover of size $F$ in the graph).

On the other hand, suppose we determined that $(\MM_F^G)$ is empty, which means we got to an iteration $i_{max}$ where $\vol(P_i) < \volumelb$. We argue that this actually means that the polyhedron $P_i$ is empty. That is, we argue that we have found a set of edges $E'\subseteq E$ which contain a matching of size $F+1$ ($E'$ is the set of edges whose constraints we added to $P_{i_{max}}$ during the cutting planes method). If this was not the case, that is if the maximum matching size in $E'$ is at most $F$, then
it must be the case, by \cref{Claim:Konig_claim}, that a vertex cover of size $F$ exists in the subgraph $G' = (L\cup R, E')$, and hence that some integer point exists in our polyhedron $P_{i_{max}}$. We can deduce that this is impossible, however, by simply noting that by the behaviour of the algorithm, it holds that $(\MM_F^{G'})=P_{i_{max}}$, and thus we can apply \cref{claim:volume_lower_bound} which then says that $\vol(P_i)\geq \volumelb$, which is a contradiction.
\end{proof}

By \cref{clm:fvedge} and \cref{claim:number_of_iterations_bound} we see that the algorithm makes a total of $O(n\log^2 n)$ $\ORQ$-queries, and \cref{lemma:correctness_meta_algorithm} argues its correctness. This concludes the proof of \cref{lemma:OR_query_result}.

\subsection{$\ORQ$-query optimization algorithm}\label{ssec:optimization_alg}
In this section we describe a standard modification (see e.g.~\cite[Section 4]{Vaidya89}) to our cutting planes \emph{decision} algorithm, so that it solves the \emph{optimization} version with the same query-complexity.

\begin{lemma}\label{lemma:or_bmm}
There is a deterministic algorithm using $O(n\log^2 n)$ $\ORQ$-queries which solves the $\BMM$ problem. In particular, the algorithm finds a maximum matching $M$, together with a witness that $M$ is maximum in the form of a fractional vertex cover of size strictly less than $|M|+1$.
\end{lemma}

\begin{proof}
The idea is to run \cref{alg:two} starting with $F = 2n$. Whenever the algorithm finds a feasible point $p_i$, instead of terminating, we lower the value of $F$ instead. The point $p_i$ is a certificate that a vertex cover of size $\left\lfloor \sum_v p_i^v \right\rfloor$ exists (since \eqref{eq:MM} is integral). Hence we lower $F$ to $F\gets \left\lfloor \sum_v p_i^v \right\rfloor - 1$, by adding the constraint $\sum_v x_v \le F+\frac{1}{3}$, and continue the cutting planes algorithm.
Note that the constraint $\sum_v x_v \le F+\frac{1}{3}$ forms a \emph{violating constraint} for $p_i$ (and therefore cuts down the volume by a constant fraction, see \cref{lem:volume-reduction}, and counts as an iteration of the cutting planes algorithm).

At the end, the algorithm must terminate by determining that $(\MM_F^G)$ is empty (for the current value of $F$), in which case the found edges $E'$ contains a matching of size $F+1$ (see \cref{lemma:correctness_meta_algorithm}). On the other hand, the last time we lowered $F$, we had a fractional vertex cover $p_i$ of size strictly less than $F+2$.
\end{proof}

\subsection{Applications}
\label{ssec:applications}

The goal of this section is to complete the proof of \cref{thm:BPM_in_all_models}. We prove the theorem by showing how to simulate the $\ORQ$-query cutting planes algorithm in the communication setting and the different query models (randomized $\XORQ$, $\ISQ$, $\ORQ_k$, and quantum edge query). 

\subsubsection{Communication complexity} \label{sec:appl-cc}
We first consider the two-party edge-partition communication setting, where the edges $E$ of the graph are partitioned into sets $E_A$ and $E_B$ given to Alice and Bob respectively.

\begin{claim}
There is a communication protocol solving $\BMM$ in $O(n\log^2 n)$ bits of communication.
\end{claim}

A standard way of doing this is to simulate each $\ORQ$-query $S\subseteq L\times R$ with 2 bits of communication: Alice and Bob check locally if $S\cap E_A$, respectively $S\cap E_B$, is non-empty and then share this information with each-other.

Alternatively, Alice and Bob can implement the ``$\fvedge$''-subroutine of \cref{alg:two} directly by checking locally for a violating edge and sharing it, if they find one, to the other party. This makes sure that $E'$ is mutually known throughout the protocol. Sending an edge requires $O(\log n)$ bits of communication, and needs to be done $O(n\log n)$ times. So this alternative approach achieves the same final communication complexity (although in slightly fewer rounds of communication), and is also closer to our weighted matching algorithm in \cref{sec:weighted-cp} (where the query-settings are no longer compatible).

\subsubsection{Randomized $\XORQ$-query} \label{sec:appl-xor}
 Now we turn to the $\XORQ$-query setting. \cite{BeniaminiN21} showed that solving $\BPM$ is \emph{evasive} for the $\XORQ$-query setting for any \emph{deterministic} algorithm, meaning that any such algorithm needs to make $n^2$ queries (that is, the trivial algorithm for querying every potential edge individually is optimal)! Nevertheless, we show that \emph{randomized} $\XORQ$-query algorithms are much more powerful, and can achieve almost linear number of queries instead.

\begin{claim}\label{claim:XOR_query_result}
There is a randomized algorithm which makes $O(n \log^2 n)$ $\XORQ$-queries and, w.h.p.\footnote{\emph{w.h.p. = with high probability}; meaning with probability at least $1-1/n^{c}$ for an arbitrarily large constant $c$.}, solves $\BMM$.
\end{claim}

In order to establish this result, we need the following folklore observation.
\begin{observation}\label{obs:sample_gives_good_xor}
For any $k$, let $x\in \{0,1\}^k$ be a binary string of length $k$, such that $x \neq 0^k$. If $r\in \{0,1\}^k$ is picked uniformly at random, then $\Pr\left[(\sum_{i=1}^{k} x_i r_i) \text{ is odd}\right] = \frac{1}{2}$.
\end{observation}

\begin{lemma} \label{lem:simul-xor-or}
A single $\ORQ$-query can be simulated, w.h.p., by issuing $O(\log n)$ randomized $\XORQ$-queries.
\end{lemma}
\begin{proof}
If we want to simulate an $\ORQ$-query over a subset $S$, we can sample $S'\subseteq S$ randomly (independently keep every element with probability $\tfrac{1}{2}$) and issue an $\XORQ$-query over $S'$. If the answer to said $\ORQ$-query was ``YES'', then we have, by \cref{obs:sample_gives_good_xor}, a constant probability of realizing this with our $\XORQ$-query over $S'$. If we repeat $O(\log n)$ times, we can answer the $\ORQ$-query correctly w.h.p.
\end{proof}

\begin{proof}[Proof of \cref{claim:XOR_query_result}.]
Just applying \cref{lem:simul-xor-or} to our $\ORQ$-query algorithm would imply an $O(n \log^3 n)$ randomized $\XORQ$-query algorithm.
An additional observation is required to bring the query complexity down to $O(n\log ^2 n)$. We note that in each invocation of $\fvedge$, we need only simulate the first $\ORQ$-query, after which we, w.h.p., have in hand a concrete set $S'\subseteq S$ for which $\XORQ(S')=1$ (or else determined that the answer to said $\ORQ$-query should be ``NO''). At this point we can binary-search \emph{deterministically} using an additional $O(\log n)$ $\XORQ$-queries to find a violating edge in $S'$. Hence, each invocation of $\fvedge$ can be simulated, w.h.p., via $O(\log n)$ $\XORQ$-queries; and thus by \cref{claim:number_of_iterations_bound,lemma:correctness_meta_algorithm}, the entire algorithm requires $O(n\log^2 n)$ $\XORQ$-queries and is correct w.h.p.
\end{proof}

\subsubsection{Independent set ($\ISQ$) query} \label{sec:appl-is}
In this section we discuss a restricted version of the $\ORQ$-query, namely the \emph{Independent Set} ($\ISQ$) query, as studied by, for example, \cite{BeameHRRS18,AuzaL21, RashtchianWZ20, AlonA05, AlonBKRS04, AbasiN19}. An $\ISQ$-query consists of specifying two subsets $X\subseteq L$ and $Y\subseteq R$ and asking if there is any edge between some vertex in $X$ and some vertex in $Y$ (or, conversely if $X\cup Y$ forms an independent set)\footnote{Generally, an Independent set query specifies only one subset of vertices whereas the \textit{Bipartite} independent set query specifies two disjoint sets of vertices as defined here. However, for bipartite graphs, these two types of queries are equivalent.}.

\begin{claim}\label{claim:BMM_ORn}
There is a deterministic algorithm which solves $\BMM$ with $O(n \log^2 n)$ $\ISQ$-queries.
\end{claim}
\begin{proof}
In each iteration, the cutting plane method finds some fractional point $p\in \bbR^{L\cup R}$, and we are asked to implement a separation oracle $\fvedge$ for this point. That is we want to determine if any edge in the set $S = \{(u,v)\in L\times R \mid p_u + p_v < 1\}$ exists (and if so find it). With unrestricted $\ORQ$-queries this is easy (see \cref{clm:fvedge}), however it might not be the case that this set $S$ is structured like an $\ISQ$-query. In the case when $p$ is integral,
we can define $X = \{v\in L : p_v = 0\}$
and $Y = \{v\in R : p_v = 0\}$, and note that $S = X\times Y$. Hence, in the case of integral $p$, we can implement $\fvedge$ using $\ISQ$-queries: first we binary search on $X$, and then on $Y$, to find the violating edge if it exists.

We argue that there always exist an integral point $p'\in \bbZ^{L\cup R}$ which we can use instead of $p$ when calling the separation oracle $\fvedge$. The integral point $p'$ will satisfy the following two properties:
\begin{enumerate}[(i)]
    \item  For all pairs $(u,v)\in L\times R$, if $p_u + p_v \ge 1$ then $p'_u + p'_v \ge 1$ too. This means that if we found a violating edge for $p'$, the same edge is also violating for $p$.
    \item $\sum p'_v \le \sum p_v$. This means that if there were no violating edges (i.e. $p'$ formed a vertex cover), we have found a certificate that the maximum matching size is at most $\sum p'_v \le \sum p_v$.
\end{enumerate}
Indeed, consider the bipartite graph $H$ with edge set $\{(u,v)\in L\times R : p_u + p_v \ge 1\}$. In $H$, $p$ is a (fractional) vertex cover of size $\sum p_v$. This means that there exists an integral vertex cover of size $\lfloor\sum p_v\rfloor$ in $H$, since the minimum vertex cover linear program is integral for bipartite graphs.
Therefore, we pick $p'$ to be an arbitrary such integral vertex cover, and we note that by definition it satisfies the above properties (i) and (ii).
\end{proof}

\subsubsection{$\ORQ_k$-query} \label{sec:appl-ork}

Here we discuss the $\ORQ$-query of \emph{limited width} $k$, i.e.\ the $\ORQ_k$-query. That is, we are only allowed to ask $\ORQ$-queries over sets $S\subseteq L\times R$ of size $|S|\le k$. This model turns out to be useful as an intermediary step towards proving tight upper bounds for the quantum edge query model (see \cref{sec:quantum}). Considering this model also helps to unveil the difficulty behind designing \emph{demand query} algorithms (see open problems in \cref{sec:open-problems} for further discussion) for $\BPM$, pointing to the fact that the barrier is not the size of the query, but rather its locality.

\begin{claim}\label{claim:BMM_ORn}
There is a deterministic algorithm which solves $\BMM$ with $O(n \log^2 n)$ $\ORQ_n$-queries.
\end{claim}

In fact, we show, via an amortization argument, that \emph{any} $\ORQ$-query algorithm (for an arbitrary graph problem) can be simulated with $\ORQ_k$-queries with an additive overhead dependent on $k$.

\begin{lemma}\label{lemma:bounded_or_query_simulation}
Any $\ORQ$-query algorithm $\mathcal{A}$ (for any graph problem) making $q$ queries can be converted to an $\ORQ_k$-algorithm making $q+\lfloor \tfrac{n^2}{k}\rfloor$ queries.
\end{lemma}
\begin{proof}
The main idea of the proof is the following: Every time an $\ORQ$-query answers ``NO'', we know that none of the queried edges are present in the graph $G$. This is an important piece of information that helps us
save queries in the future. More formally, we use the following amortization argument.

We keep track of a set $E^c\subseteq L\times R$ of pairs $(u,v)$ which we know are \textbf{not} edges of $G$, that is $E\cap E^{c} = \emptyset$. Whenever $\mathcal{A}$ issues an $\ORQ$-query $S\subseteq L\times R$ we do the following. Let $S_1, S_2, \ldots, S_r$ be a partition of $S\setminus E^{c}$ so that $|S_1| = |S_2| = \ldots |S_{r-1}| = k$, and $|S_r| \le k$.
We issue the $\ORQ_k$-queries $S_1, S_2, \ldots, S_r$ sequentially, in order, until we get a ``YES'' answer which we return to $\mathcal{A}$ (or else, after we have received ``NO'' from all the sets we return a ``NO'' answer to $\mathcal{A}$). For the last query which we made (which was either a query to $S_r$, or a query which returned ``YES''), we charge the cost to $\mathcal{A}$. In total, we thus charge at most $q$ cost to $\mathcal{A}$: one per $\ORQ$-query $\mathcal{A}$ issues.

For all the queries to $S_i$ which got ``NO'' answers and for which $i < r$, we update $E^{c} \gets E^{c} \cup S_i$. Hence, for each such ``NO'' answer we have increased the size of $E^{c}$ by $k$. Note that this can happen at most $\lfloor \tfrac{n^2}{k} \rfloor$ times.
So, other than the $q$ queries charged to $\mathcal{A}$, we have made at most $\lfloor \tfrac{n^2}{k} \rfloor$ queries.
Hence, in total, we made $q+\lfloor \tfrac{n^2}{k} \rfloor$ $\ORQ_k$-queries to simulate the $q$ many $\ORQ$-queries from $\mathcal{A}$.
\end{proof}

Plugging in our $O(n\log ^2n)$ $\ORQ$-query algorithm to the above lemma yields \Cref{claim:BMM_ORn}.

\subsubsection{Quantum edge query $(Q_2)$}
\label{sec:quantum}

In this section we consider the quantum \textbf{edge}-query model. See \cite{BuhrmanW02} for a formal definition and \cite{NCQuantumBook} for a more extensive background on quantum computing. 

\begin{claim}
\label{clm:quantum}
The quantum edge query complexity of solving $\BMM$ is $O(n\sqrt n \log^2 n)$.
\end{claim}

We use our $\ORQ_k$-query algorithm (\cref{lemma:bounded_or_query_simulation}) together with a well known quantum result (\cref{lem:grover}) regarding the quantum query complexity of the $\ORQ$ function.

\begin{lemma}[{Grover Search \cite{Grover96}}]\label{lem:grover}
There is quantum query algorithm that computes w.h.p. the $\ORQ$ function over $k$ bits with query complexity $O(\sqrt{k}\log k)$. \end{lemma}

\begin{proof}[Proof of \cref{clm:quantum}.]
Consider the instance of \cref{lemma:bounded_or_query_simulation} where we put $k=\frac{n}{\log^2 n}$, then we obtain an $\ORQ_k$-query algorithm for $\BMM$ with $O(n\log^2 n)$ queries. Each such $\ORQ_k$-query can be simulated, w.h.p., using $O(\sqrt{k}\log n) = O(\sqrt{n})$ quantum edge queries, by \cref{lem:grover}.
\end{proof}

%% file: max-cost-b-matching.tex
\section{Weighted and Vertex-Capacitated Variants}
\label{sec:weighted}

In this section we show that the cutting planes method is strong enough to be generalized to solve weighted and (vertex-)capacitated problems, for example \emph{max-cost $b$-matching}.
As an application, we also show how to solve \emph{unique bipartite perfect matching} $(\UBPM)$ in \cref{sec:ubpm}. For the weighted problems, we focus on the two-party edge-partition communication setting, since there is no natural generalization of the $\ORQ$-queries.

 \begin{theorem}
 \label{thm:weighted}
 Given that all the weights/costs/capacities are integers polynomially large in $n$, we can solve the following problems\footnote{See \cite[Section~8.6 (full version)]{BrandLNPSSSW20} for a more extensive discussion of these variants.} in the two-party edge-partition communication setting, using $O(n\log^2 n)$ bits of communication.
 \begin{enumerate}[(i)]
     \item Maximum-cost bipartite perfect $b$-matching.
     \item Maximum-cost bipartite $b$-matching.
     \item Vertex-capacitated minimum-cost $(s,t)$-flow.
     \item Transshipment (a.k.a.~uncapacitated minimum-cost flow).
     \item Negative-weight single source shortest path.\footnote{Although the reduction shown in \cite{BrandLNPSSSW20} is randomized, we note that it can be made deterministic by noticing that both parties in the end will know all the edges of a shortest path tree.} 
     \item Minimum mean cycle.
     \item Deterministic Markov Decision Process (MDP).
 \end{enumerate}
 \end{theorem}

\paragraph{Reductions.}
Problems (i), (ii), (iii), (iv) are equivalent, and problems (v), (vi), (vii) can all be reduced to, for example, (i). All these reduction are  shown in \cite{BrandLNPSSSW20} (and can be verified as rectangular reductions, i.e., compatible with the two-party communication setting), \emph{except} that (ii) (i.e.\ max-cost, not-necessarily-perfect, bipartite $b$-matching) can solve any of (actually really all of): (i), (iii), (iv); which we show in \cref{sec:perfect-reduction}.

These reductions allow us to focus on a single one of these problems. We pick item (ii), that is \emph{max-cost (not-necessarily-perfect)\footnote{Some of the other problems, e.g.\ \emph{perfect} $b$-matching, can have unbounded dual linear programs which make them trickier to work with in the cutting planes framework.} bipartite $b$-matching}, which is the one which most closely resembles the unweighted bipartite matching problem. In \cref{sec:weighted-cp} we show how the cutting planes framework can be generalized to work with the costs $c$ and demand vector $b$.

\begin{definition}[$b$-matching]\label{def:b-matching}
Given a graph $G = (V, E)$, a demand vector $b\in \bbZ_{\ge 0}^V$, and edge-costs $c \in \bbZ^E$,
we call a vector
$y\in \bbZ^{E}_{\ge 0}$
a \emph{$b$-matching}
(or a \emph{fractional} $b$-matching if we allow $y\in \bbR^{E}_{\ge 0}$) if
$\sum_{e\in \delta(v)} y_e \le b_v$ for all $v\in V$ (where $\delta(v)$ is the set of edges incident to $v$).
If
$\sum_{e\in \delta(v)} y_e = b_v$ for all $v\in V$, then $y$ is a \emph{perfect} $b$-matching.
The \emph{cost} (or \emph{weight}) of $y$ is $\sum_{e\in E} c_e y_e$.
\end{definition}

\subsection{Max-cost perfect $b$-matching $\to$ Max-cost $b$-matching}
\label{sec:perfect-reduction}

We can reduce the perfect variant to the not-necessarily-perfect one. Suppose we are given an instance $(G=(V,E), b\in\bbZ^{V}_{\ge 0},c\in \bbZ^{E})$ of the perfect variant which we wish to solve.
Firstly, we may assume that the costs are non-negative, since adding a constant $W$ to all costs will increase the cost of a \emph{perfect} $b$-matching by exactly
$W\frac{\sum_{v\in V} b_v}{2}$.

If we just solve max-cost $b$-matching, we in general do not obtain a perfect $b$-matching, since matchings of smaller cardinality might have higher cost. To encourage the max-cost $b$-matching to prioritize perfect matchings over non-perfect matchings, we simply add a large integer $W$ to all the the costs. That is we use the cost function $c'_e = c_e + W$ instead (again, we can do this since we know exactly how this will affect the cost of a perfect $b$-matching). If $W$ is sufficiently large (in particular set $W \coloneqq 1 + |V|\cdot \max c_e$), any max-cost $b$-matching will also be a perfect $b$-matching (if one exist).

If it was not, suppose $M$ is a non-perfect $b$-matching and of maximum cost for $c'$, in a graph which allows a perfect $b$-matching. Then there must exist an augmenting path in $M$ of length $2\ell+1\le |V|$ (where we add $\ell+1$ edges and remove $\ell$). The total cost (w.r.t.\ $c'$) of the added edges is now at least $(\ell+1)W$, while the cost of the removed edges is at most
$\ell(W+\max c_e) \le \ell W + |V| \max c_e < \ell W + W = (\ell+1)W$. That is we added more cost than we removed, hence contradicting that $M$ was of maximum cost.

\subsection{Cutting planes method for max-cost bipartite $b$-matching}
\label{sec:weighted-cp}

In this section we briefly explain how the cutting planes algorithm can solve the \emph{max-cost bipartite $b$-matching}, and hence prove \cref{thm:intro-other-problems}. The details are postponed to \cref{apx:weighted}. The main result of this section is the following:

\begin{lemma}\label{lemma:max_cost_b_matching_result}
Max-cost bipartite $b$-matching can be solved using $O(n \log^2 (nW))$ communication, where $W := \max \{\max |c_e|, \max b_v\}$ is the largest number in the input.
\end{lemma}

Let $G = (V, E)$ (bipartite with $|V| = n$), $b\in \bbZ^{V}_{\ge 0}$ and $c\in \bbZ^{E}$ be an instance of the max-cost bipartite $b$-matching problem. We assume the edges (together with their costs) are partitioned between two parties Alice and Bob, say Alice owns $E_A$ (together with $c_e$ for $e\in E_A$) and Bob $E_B$ (together with $c_e$ for $e\in E_B$). We assume both players know the demands $b$ (otherwise it can be communicated in $O(n \log W)$ bits).

\paragraph{Dual linear program.}
Similarly as for the unweighted bipartite matching problem, we run a cutting planes algorithm on the dual linear program \eqref{eq:WMM} (refer to the constraints of \cref{def:b-matching} for the primal linear program). We can think of $x\in \eqref{eq:WMM}$ as a generalized version of a vertex cover, and an optimal solution would be one of minimum cost (w.r.t.\ costs $b_v$). Similarly to the uncapacitated and unweighted case, since the graph is bipartite, \eqref{eq:WMM} is integral, and thus has an \emph{integral} optimal solution.

\begin{equation}
\begin{array}{ll@{}ll}
\text{min} & \displaystyle\sum\limits_{v\in V} b_v x_{v}  &\\
\text{s.t.}&\displaystyle x_u+x_v\geq c_{uv}   && \forall (u,v)\in E\\
                            &x_{v} \ge 0
                            && \forall v\in V
\end{array}
\tag{$\MM^{(G,b,c)}$}\label{eq:WMM}
\end{equation}

\begin{claim}\label{claim:optimal_soultion_below_W}
Any optimal solution $x^*$ to \eqref{eq:WMM} has $x^*_v \le W$ for all $v\in V$.
\end{claim}
\begin{proof}
If this is not the case, we can decrease $x^{*}_v$ without violating any of the constraints, and the objective value $\sum b_v x_v$ becomes smaller.
\end{proof}

This motivates the following feasibility polytope \eqref{eq:WMMF}, which can be used to check if \eqref{eq:WMM}  has a solution with objective value at most $F$, for any integer $F\in \bbZ$.

\begin{equation}
\begin{array}{ll@{}ll}
\text{} & \displaystyle\sum\limits_{v\in V} b_v x_{v} \le F+\tfrac{1}{3}  &\\
\text{}&\displaystyle x_u+x_v\geq c_{uv}   && \forall (u,v)\in E\\
                            &0\le x_{v} \le W+1
                            && \forall v\in V
\end{array}
\tag{$\MM^{(G,b,c)}_F$}\label{eq:WMMF}
\end{equation}

\paragraph{Modifications to the algorithm.}
Like for the unweighted and uncapacitated case,  we can show that if this polytope is non-empty, then it has significantly large volume (\cref{lem:volume_lower_bound_weighted}, whose proof is in \cref{apx:weighted}). This means that the cutting plane algorithm can terminate whenever the volume becomes too small.
The only modifications we need to make to \cref{alg:two} are thus the following:
\begin{itemize}
\item We start with a larger initial polytope $P_0 = [0,W+1]^{V} \cap \{x\in \bbR^{V} : \sum b_v x_v \le F+\tfrac{1}{3}\}$.

\item When we check for (and add) violating constraints we also use the edge-cost $c_{uv}$. That is an edge $(u,v,c_{uv})$ is violating if $p^u_i + p^v_i < c_{uv}$, and the corresponding constraint we add is ``$x_u + x_v \ge c_{uv}$''.

\item We terminate when the volume is less than $(\tfrac{1}{20nW})^{n}$ (see \cref{lem:volume_lower_bound_weighted}).
\end{itemize}

\begin{restatable}[Generalization of \cref{claim:volume_lower_bound}]{lemma}{weightedvolumebound} \label{lem:volume_lower_bound_weighted}
If $\eqref{eq:WMMF}$ is non-empty, then $\vol\eqref{eq:WMMF}\ge\left(\frac{1}{20nW}\right)^{n}$.
\end{restatable}

\begin{observation}[Generalization of \cref{clm:fvedge}]
\label{obs:fvedgew}
We can communicate a single violated constraint with $2\log(n)+\log(W)+1 = O(\log(nW))$ bits of communication.
\end{observation}

\paragraph{Total communication.} The generalized algorithm will start with initial polytope $P_0 \subseteq [0,W+1]^{n}$, and terminate whenever $\vol(P_i)$ becomes smaller than $\left(\frac{1}{20nW}\right)^{n}$ (\cref{lem:volume_lower_bound_weighted}). In each iteration the volume is cut down by a constant fraction (\cref{lem:volume-reduction}). Hence we need $O\left(\log \left(W^n / \left(\frac{1}{20nW}\right)^{n}\right)\right) = O(n \log (nW))$ iterations. Each iteration needs $O(\log(nW))$ bits of communication, for a total of $O(n\log^2 (nW))$ bits of communication (\cref{obs:fvedgew}), proving \cref{lemma:max_cost_b_matching_result}. We also note that the same standard trick to convert the \emph{decision} version to the \emph{optimization} version (\cref{ssec:optimization_alg}) works here as well.

\subsection{Application: Unique Bipartite Perfect Matching}
\label{sec:ubpm}

In the \emph{unique bipartite matching problem}, or $\UBPM$ for short, we are asked to determine if an (unweighted) bipartite graph has a \emph{unique} perfect matching. The interplay between $\UBPM$ and $\BPM$ is quite subtle, by some measures, e.g.~certificate complexity, the former is known to be strictly harder than the latter, in other settings, such as sequential, simple near linear time algorithms for $\UBPM$ have been known since the turn of the century \cite{GabowKT01}.
While for $\BPM$, only a very recent line of work, employing heavy machinery from continuous optimization and dynamic data structures culminated in a near linear time algorithm for $\BPM$ \cite{ChenKLPGS22}.
In this section, we show that our upper bounds also hold for the $\UBPM$ problem, both for the communication and query models.

\begin{restatable}{theorem}{theoremUBPM}
\label{thm:UBPM}
The $\UBPM$ problem can be solved in:
\begin{itemize}
\item $O(n\log^2 n)$ bits of communication in the deterministic two-party edge-partition communication model.
\item $O(n\log^2 n)$ deterministic OR-queries.
\item $O(n\log^2 n)$ randomized XOR-queries, w.h.p.
\item $\tO(n\sqrt{n})$ quantum edge queries.
\end{itemize}
\end{restatable}

\begin{proof}[Proof sketch.] (Formal proof can in \cref{apx:weighted}).\\
Our main idea on how to solve $\UBPM$ can be summarized in two stages:
\begin{enumerate}[(1)]
    \item First find a perfect matching $M$ (see \cref{thm:BPM_in_all_models}).
    \item Assign weights to the edges as follows: $c_e = 1$ if $e\in M$, and $c_e = 2$ otherwise. After this, we find a \emph{max-cost} perfect matching
    $M'$ (see \cref{thm:weighted}).
\end{enumerate}
If $M' \neq M$, we have proved that the perfect matching is not unique.
Conversely, if $M' = M$, then $M$ must be the \emph{unique} perfect matching, since any other perfect matching (if they would exist) has higher cost.
In the communication setting, this argument suffices. For the query models, however, one needs to be a bit more careful since we have not defined what, for example, an $\ORQ$-query means in the \emph{max-cost} setting. The formal proof of this---which is straightforward, although a bit technical---can be found in \cref{apx:weighted}. The other query-models follow from similar reductions as those in \cref{ssec:applications}.
\end{proof}

\begin{remark}
We also note that there are alternative ways to solve $\UBPM$ after one has solved $\BPM$. Instead of solving \emph{max-cost matching} as the second step, one can instead determine if an alternating cycle (that is a cycle in $G$ where every other edge is in $M$) exist. The perfect matching $M$ is \emph{unique} if and only if no such cycle exist. Note that finding such a cycle can easily be done in, for example, $O(n \log n)$ $\ORQ$-queries by running depth-first-searches to detect a directed cycle in the directed residual graph (edges in $M$ go from $R\to L$, other edges from $L\to R$).
\end{remark}

%% file: communication_lower_bounds.tex
\section{Communication Lower Bounds}
\label{sec:communication-lower-bounds}

In this section we discuss three communication problems related to $\ORQ$-queries, $\ANDQ$-queries, and $\XORQ$-queries respectively. We show simple lower bounds on the communication complexity of these three problems, and argue that this implies corresponding query lower bounds.

 We summarize our obtained query lower bounds below in \cref{thm:query-lb}. Note that our lower bounds are asymptotically the same as those obtained in \cite{BeniaminiN21}. However, while \cite{BeniaminiN21} employs rather sophisticated mathematical machinery in order to obtain these lower bounds, we obtain them via simple communication complexity reductions from well known functions. In addition to the already known lower bounds of \cite{BeniaminiN21}, our technique also shows one new result: namely that the $\Omega(n^2)$ $\ANDQ$-query lower bound even holds for \textbf{randomized} algorithms.  
 
 \begin{theorem}
 \label{thm:query-lb}
 To solve the bipartite perfect matching ($\BPM$) problem one needs to use:
 \begin{itemize}
 \item $\Omega(n\log n)$ $\ORQ$-queries (deterministic).
 \item $\Omega(n^2)$ $\ANDQ$-queries (deterministic or randomized).
 \item $\Omega(n^2)$ $\XORQ$-queries (deterministic).
 \end{itemize}
 \end{theorem}
 
 In this section we give an overview of the main ideas---which are all relatively simple---but we postpone the formal proofs to \cref{apx:communication-lower-bounds}.

\paragraph{Problem setup.}
We consider a two-party communication setting between two players Alice and Bob. Alice is given a graph $G_A = (V, E_A)$ and Bob a graph $G_B = (V,E_B)$ on the same set of vertices $V = L\cup R$ (where $|L|=|R|=n$).
They wish to solve $\BPM$ on an aggregate of their graphs. We consider three different types of aggregate graphs (and hence get three different communication problems), naturally corresponding to $\ORQ$ / $\ANDQ$ / $\XORQ$:

\begin{itemize}
\item Union graph $G_\cup = (V, E_A\cup E_B)$.
\item Intersection graph $G_\cap = (V, E_A\cap E_B)$.
\item Symmetric difference graph $G_\oplus= (V, E_A\oplus E_B)$.
\end{itemize}
It is not difficult to see that any $\ORQ$ / $\ANDQ$ / $\XORQ$ query algorithms can be simulated in an communication protocol for $G_\cup$ / $G_{\cap}$ / $G_{\oplus}$ respectively. As an example, to answer an $\ANDQ$-query over $S\subseteq (L\times R)$ in $G_{\cap}$, Alice and Bob check locally if $S\subseteq E_A$, respectively if $S\subseteq E_B$, and exchange this information with each-other.
This gives the following  \cref{lem:query-to-communication}, which we formally prove in \cref{apx:communication-lower-bounds}.

\begin{restatable}{lemma}{qc}
\label{lem:query-to-communication}
If there is a query-algorithm $\mathcal{A}$ solving the bipartite-perfect-matching problems using $q$ many $\ORQ$ / $\ANDQ$ / $\XORQ$ queries, then there is a communication protocol which solve the bipartite-perfect-matching problem on $G_\cup$ / $G_\cap$ / $G_\oplus$ respectively, using $2q$ bits of communication.
If $\mathcal{A}$ is deterministic, then so is the communication protocol.
\end{restatable}

\begin{figure}[h!]
\center
\includegraphics[width=0.7\textwidth]{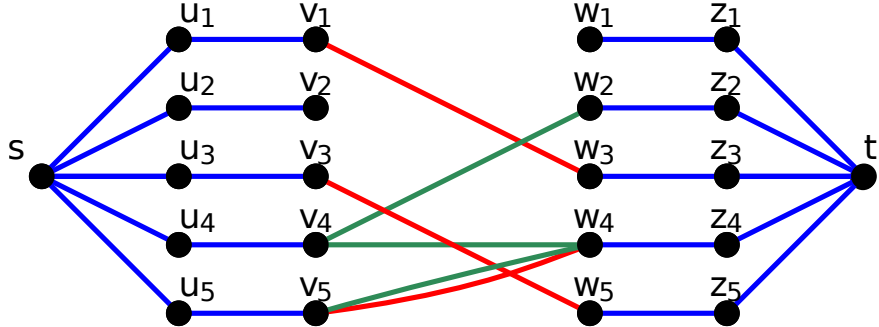}
\caption{An example of our graph construction for $G_{\cap}$. The blue edges are known to both parties to be in the aggregate graph $G_{\cap}$. Between the $v$ and $w$ layers, Alice owns the red edges and Bob the green. The graph $G_{\cap}$ has a perfect matching if and only if it has an edge between some $v_i$ and $w_j$.}
\label{fig:com-lb}
\end{figure}

\paragraph{Intersection Graph ($\ANDQ$).}
We construct a difficult instance for solving $\BPM$ on $G_\cap$ by reducing from Set-Disjointness on $\Theta(n^2)$ bits. Our construction can be seen in \cref{fig:com-lb}, where determining if $G_\cap$ has a perfect matching boils down to determining if there exists any edge between some vertex $v_i$ and some vertex $w_j$ in $G_{\cap}$. This is exactly a Set-Disjointness problem of size $\Theta(n^2)$.

Since Set-Disjointness is known to require linear communication in the number of bits \cite{KNBookCC} (both for deterministic and randomized algorithms), we obtain the following
\cref{lem:and-lower-bound} 
whose formal proof can be found in \cref{apx:communication-lower-bounds}.

\begin{restatable}{claim}{andlowerbound}
\label{lem:and-lower-bound}
Solving $\BPM$ (or $\UBPM$) on $G_\cap$ requires $\Omega(n^2)$ bits of communication, even if public randomness is allowed.
\end{restatable}

\paragraph{Symmetric Difference Graph ($\XORQ$).}
Our communication lower bound of $G_\oplus$ is very similar to the lower bound on $G_{\cap}$. We use the same graph-structure (\cref{fig:com-lb}), but determining if there is any $(v_i, w_j)$ edge in $G_{\oplus}$ now corresponds to solving the Equality problem (again on $\Theta(n^2)$ bits) instead of the Set-Disjointness problem. This is since an edge $(v_i, w_j)$ exists if and only if exactly one of Alice or Bob have it, which is if and only if the set of Alice's edges does not equal the set of Bob's edges.

It is well known that Equality exhibits a large gap between its deterministic and randomized communication complexity. While the former is known to be $\Omega(n)$, the latter requires only $O(1)$ bits of communication in the presence of shared randomness (with error probability $< \tfrac{1}{3}$) \cite{KNBookCC}. This might also provide some intuition why we, in the case of $\XORQ$-queries, obtain the separation of $\tO(n)$ randomized upper bound (\cref{claim:XOR_query_result}) vs $\Omega(n^2)$ deterministic lower bound. The full proof of \cref{lem:xor-lower-bound} can be found in \cref{apx:communication-lower-bounds}.

\begin{restatable}{claim}{xorlowerbound}
\label{lem:xor-lower-bound}
Solving $\BPM$ on $G_\oplus$ requires $\Omega(n^2)$ bits of communication for any deterministic protocol.
\end{restatable}

\paragraph{Union Graph ($\ORQ$).}
Our lower bound for $G_{\cup}$ follows a similar vein, but the graph construction is a bit different. By a standard reduction, the $(s,t)$-reachability problem on a (not-necessarily-bipartite) $n$-vertex graph can be solved by solving bipartite perfect-matching on a graph on $2n$ vertices.
The $(s,t)$-reachability problem, in the edge-partition setting, is known to need $\Omega(n \log n)$ bits of communication for any deterministic protocol \cite{HajnalMT88}. Proving any super-linear \emph{randomized} lower bound still remains open. The full proof of \cref{lem:or-lower-bound} can be found in \cref{apx:communication-lower-bounds}.

\begin{restatable}{claim}{orlowerbound}
\label{lem:or-lower-bound}
Solving $\BPM$ on $G_\cup$ requires $\Omega(n \log n)$ bits of communication for any deterministic protocol.
\end{restatable}

%% file: appendix_maxcost.tex
\section{Omitted proofs from \cref{sec:weighted}}
\label{apx:weighted}

\weightedvolumebound*

\begin{proof}

Similarly to \cref{claim:volume_lower_bound}, let $x$ be some integral feasible solution to \eqref{eq:WMMF}, we know such a solution exists since \eqref{eq:WMM} has an optimal integer solution, and by assumption that \eqref{eq:WMMF} is not empty, we can deduce that its value is at most $F+\frac{1}{3}$, which is at most $F$ since it is integer. For the same reason, we can also, by \cref{claim:optimal_soultion_below_W}, assume w.l.o.g. that $x_v\leq W$ for all $v\in V$. Our goal is to prove that \eqref{eq:WMMF} contains a cube of dimensions $\frac{1}{20nW}$, thus concluding the proof. Note that for each $v\in V$ we can independently increase the value of $x_v$ by any value in the range $[0,\frac{1}{20nW}]$ while maintaining the feasibility of the resulting point w.r.t. \eqref{eq:WMMF}. This is due to the following observations.
\begin{itemize}
    \item The constraint $\sum\limits_{v\in V}b_vx_v\leq F+\frac{1}{3}$ remains valid as as $x$ has value $F$, and we increase each coordinate by at most $\frac{1}{20nW}$, as there are $2n$ vertices and $b_v\leq W$ for all $v\in V$, it means that we increase the value of the solution by at most $\frac{1}{10}$, which is less than $\frac{1}{3}$.
    \item All edge constrains $x_v+x_u\geq c_{uv}$ for all $(u,v)\in E$ clearly remain valid as we are only increasing the value of variables.
    \item As \ref{claim:optimal_soultion_below_W} allows us to assume that $x_v\leq W$ for all $v\in V$ for the original $x$, the constraints $0\leq x_v\leq W+1$ for all $v\in V$ remain valid as well as each entry is increased by at most $\left(\frac{1}{20nW}\right)$. \qedhere{}
\end{itemize}
\end{proof}

\vspace{0.2cm}

\theoremUBPM*

\begin{proof}
Here we give the formal proof that we can determine if an unweighted bipartite graph $G = (V,E)$ (with $V = L\cup R$, $|L| = |R| = n$) has a \emph{unique} bipartite matching in all our different models. We show that $O(n\log^2 n)$ $\ORQ$-queries suffices, and the other models follows similarly as in \cref{ssec:applications}.

We start by obtaining a maximum matching $M$ of $G$, by \cref{lemma:OR_query_result}. If $M$ is not a perfect matching, $G$ admits no perfect matching, and we are done. From now on assume that $M$ is a perfect matching.

We now construct an instance of a weighted matching.
Let $c_{e} = 10n$ if $e\in M$ and $c_e = 10n+1$ otherwise. We wish to solve the max-cost matching problem with edge costs $c$. We will soon explain how to do this with $\ORQ$-queries, but first we show how doing this completes the proof.

Let $M'$ be a max-cost matching in $G$ with respect to the costs $c$. Note that $M$ has cost $10n^2$, so we know that $M'$ has cost at least $10n^2$ too.
First we argue that $M'$ must be a perfect matching of $G$.
This is because any non-perfect matching has weight at most $(n-1)(10n+1) < 10n^2$.
\begin{itemize}
\item If $M' \neq M$, we have found two perfect matchings of the graph, and are done ($G$ does not have a \emph{unique} perfect matching).
\item If $M' = M$, we conclude that $M$ is the \emph{unique} perfect matching in the graph. This is since any other perfect matching of $G$, if they would exist, would have had strictly larger cost.
\end{itemize}

Now we return to explaining how to solve the max-cost (w.r.t.\ the cost $c$) matching using only $\ORQ$-queries. Note that we already know the weights of all (potential) edges (either $10n$ if they are in $M$, which we know; or $10n+1$ otherwise), and this is not something which needs to be found out using $\ORQ$-queries. By the above discussion, it suffices to check if any matching of weight strictly more than $10n^2$ exist or not. Hence we use a version of \eqref{eq:WMMF} with $c$ being our costs, $b$ the all-ones vector, and $F = 10n^2$, to obtain the following polytope \eqref{eq:ubpm_poly}:

\begin{equation}
\begin{array}{ll@{}ll}
& \displaystyle\sum\limits_{v\in V} x_{v} \le 10n^2+\tfrac{1}{3}  &\\
&\displaystyle x_u+x_v\geq 10n   && \forall (u,v)\in M\\
&\displaystyle x_u+x_v\geq 10n+1   && \forall (u,v)\in E\setminus M\\
                            &0 \le x_{v} \le 10n+2
                            && \forall v\in V
\end{array}
\tag{$\MM^{\UBPM}$}
\label{eq:ubpm_poly}
\end{equation}

If this polytope is feasible, then $M$ must be the unique perfect matching, and if it is feasible, another perfect matching $M'$ of higher cost exist.

We note that in the above polytope (which is a special case of the \eqref{eq:WMMF} polytope from \cref{sec:weighted-cp}), we know all the constraints initially except the
``$ x_u + x_v \ge 10n$ for $(u,v)\in E\setminus M$'' constraints. Hence, when running the cutting planes algorithm we may start with all other constraints in our initial polytope. Now, to implement the $\fvedge$-subroutine (a.k.a.\ the separation oracle) given a point $p_i$, we binary search (with $\ORQ$-queries) over the set of potentially violated edges, i.e. the set $S = \{(u,v)\in L\times R \mid p_i^u + p_i^v < 10n+1, (u,v)\not\in M\}$ (like in \cref{clm:fvedge}, but now this set $S$ is a bit different for our problem). Otherwise, the cutting planes algorithm is exactly the same as for the max-cost $b$-matching in \cref{sec:weighted-cp}.
\end{proof}

%% file: appendix_communication_lower_bound.tex
\section{Omitted proofs from \cref{sec:communication-lower-bounds}}
 \label{apx:communication-lower-bounds}
 
 \qc*
 
\begin{proof}
Alice and Bob can simulate a query on the aggregate graph with a single bit sent in each direction as follows:
\begin{itemize}
\item An OR-query $S \subseteq L\times R$ can be simulated on $G_\cup$: Alice and Bob locally check if $|S\cap E_A| > 0$, respectively $|S \cap E_B| > 0$, and communicates this to the other party.
\item An AND-query $S \subseteq L\times R$ can be simulated on $G_\cap$: Alice and Bob locally check if $|S\cap E_A| = |S|$, respectively $|S \cap E_B| = |S|$, and communicates this to the other party.
\item An XOR-query $S \subseteq L\times R$ can be simulated on $G_\oplus$: Alice and Bob locally compute the parity of $|S\cap E_A|$, respectively $|S \cap E_B|$, and communicates this to the other party. If they have the same parity, then the parity of $|S\cap (E_B\oplus E_A)|$ is even, otherwise it is odd.
\qedhere{}
\end{itemize}
\end{proof}

\andlowerbound*
\begin{proof}
Let $k$ be an integer and suppose Alice and Bob are tasked to solve a Set-Disjointness problem on $k^2$ bits.
That is, suppose Alice is given a subset $A\subseteq [k]\times [k]$ and Bob is given a subset $B\subseteq [k]\times [k]$, and they want to determine if $A\cap B$ is empty or not.
This is known to require $\Omega(k^2)$ bits of communication, even if public randomness is allowed \cite{Razborov92}.

We now proceed by setting up two bipartite graphs $G_A = (L\cup R, E_A)$ and $G_B = (L\cup R, E_B)$ such that the graph
$G_\cap = (L\cup R, E_A\cap E_B)$ has a perfect matching if and only if $A\cap B \neq \emptyset$. For an illustration, see \cref{fig:com-lb}.

The graph will have $4k+2$ vertices.
Let $L = \{s,\ v_1, v_2, \ldots, v_k,\ z_1, z_2, \ldots, z_k\}$
and 
$R = \{t,\ u_1, u_2, \ldots, u_k,\ w_1, w_2, \ldots, w_k\}$.
The edges $(s,u_i)$, $(u_i,v_i)$, $(w_i, z_i)$ and $(z_i,t)$ (for all $i\in [k]$) will be in both Alice's and Bob's graphs.
Note that the edges $(u_i, v_i)$ and $(w_i,z_i)$ form an almost-perfect-matching in $G_\cap$: all vertices except $s$ and $t$ are matched. Hence, a perfect-matching exists in $G_\cap$ if and only if there is an augmenting path (with respect to the almost-perfect-matching) between $s$ and $t$ in the graph.

Additionally, for every $(i,j)\in A$ we add the edge $(v_i,w_j)$ to Alice's edges, and similarly for every edge $(i,j)\in B$ we add the edge $(v_i,w_j)$ to Bob's edges.

\begin{itemize}
\item If $(i,j)\in A\cap B$ exists, then $(v_i,w_j)$ is an edge of $G_\cap$, and $G_\cap$ has a perfect matching:
$(s,u_i,v_i,w_j,z_j,t)$ forms the desired $(s,t)$-augmenting path.
\item On the other hand, if $A\cap B = \emptyset$, then $G_\cap$ has no perfect matching since $s$ and $t$ are in different components in the graph and hence no augmenting path between them exists.
\end{itemize}

Note that the version of Set-Disjointness where we are promised that $|A\cap B|\le 1$ is still difficult (i.e.\ also has an $\Omega(n^2)$ communication lower bound) \cite{Razborov92}. This means that our lower bound also holds for the \emph{unique} bipartite matching problem ($\UBPM$), as there can never be more than one perfect matching in this promise version.
\end{proof}

\xorlowerbound*

\begin{proof}
Let $k$ be an integer and suppose Alice and Bob are tasked to solve a Equality problem on $k^2$ bits.
That is, suppose Alice is given a subset $A\subseteq [k]\times [k]$ and Bob is given a subset $B\subseteq [k]\times [k]$, and they want to determine if $A = B$ (this is equivalent to determine if $A\oplus B = \emptyset$).
This is known to require $\Omega(k^2)$ bits of communication for any deterministic algorithm.

We use a similar graph construction on $4k+2$ vertices as in the $G_\cap$-query setting (again, for an illustration, see \cref{fig:com-lb}).
The edges $(s,u_i)$, $(u_i,v_i)$, $(w_i, z_i)$ and $(z_i,t)$ (for all $i\in [k]$) will be in be in Alice's graph, but not Bob's (so that all these edges will be $G_{\oplus}$).
Again, note that a perfect-matching exists in $G_\oplus$ if and only if there is an augmenting path (with respect to the almost-perfect-matching $\{(u_i,v_i)$, $(w_i,z_i) : i\in [k]\}$) between $s$ and $t$ in the graph.

Additionally, we add the edge $(v_i,w_j)$ to Alice's graph if $(i,j)\in A$, and similarly add the edge $(v_i, w_j)$ to Bob's graph if $(i,j)\in B$.
That is the edge $(v_i,w_j)$ is in $G_{\oplus}$ if and only if $(i,j) \in A\oplus B$.

\begin{itemize}
\item If $(i,j) \in A \oplus B$ exist, then then $(v_i,w_j)$ is an edge of $G_\oplus$, and $G_\oplus$ has a perfect matching:
$(s,u_i,v_i,w_j,z_j,t)$ forms the desired $(s,t)$-augmenting path.
\item On the other hand, if $A\oplus B = \emptyset$, then $G_\oplus$ has no perfect matching since $s$ and $t$ are in different components in the graph and hence no augmenting path between them exists.
\qedhere{}
\end{itemize}
\end{proof}

\orlowerbound*

\begin{proof}   
Suppose the two parties are given an instance of $(s,t)$-connectivity. That is Alice is given edges $F_A$ and Bob $F_B$ in a $k$-vertex (not-necessarily-bipartite) graph $H = (V,F_A\cup F_B)$. They are also both given two vertices $s,t\in V$ and want to determine if $s$ and $t$ are in the same connected component in $H$. This is known to require $\Omega(n \log n)$ bits of communication.

We proceed by constructing bipartite graphs $G_A = (L\cup R, E_A)$ and $G_B = (L\cup R, E_B)$ for Alice and Bob, such that $G_\cup = (L\cup R, E_A\cup E_B)$ has a perfect matching if and only if $s$ and $t$ are connected in $H$. 

We let $L = \{v : v\in V\setminus \{s\}\}$
and $R = \{v' : v\in V\setminus\{t\}\}$.
Note that $|L| = |R| = k-1$.
Let the edge $(v,v')$ be in both $E_A$ and $E_B$ (and thus also an edge of $G_\cup$) for all $v\in V\setminus \{s,t\}$.
These edges form an almost-perfect matching of size $k-2$: all vertices except $s$ and $t'$ are matched. Again, $G_\cup$ has a perfect matching if and only if there is an augmenting path (with respect to this partial almost-perfect matching) between $s$ and $t$ in the graph.

For each edge $(v,u) \in F_A$, we add $(v,u')$  (unless $v = t$) and $(v',u)$ (unless $v = s$) to Alice's edges $E_A$. We do the same for Bob.

Now there is an $(s,t)$-path in $H$ if and only if there is
an augmenting path (w.r.t.\ the almost-perfect matching $\{(v,v') : v\in V\setminus\{s,t\}\}$) between $s$ and $t$ in $G_{\cup}$.
Indeed, a path $(s,v_1, v_2, \ldots, v_r,t)$ in $H$ corresponds to the augmenting path
$(s,v_1',v_1, v_2', v_2, \ldots, v_r', v_r, t)$ in $G_{\cup}$.
\end{proof}

%% file: main.bbl
\newcommand{\etalchar}[1]{$^{#1}$}
\begin{thebibliography}{DHHM06}

\bibitem[AA05]{AlonA05}
Noga Alon and Vera Asodi.
\newblock Learning a hidden subgraph.
\newblock {\em {SIAM} J. Discret. Math.}, 18(4):697--712, 2005.

\bibitem[AB19a]{AbasiN19}
Hasan Abasi and Nader~H. Bshouty.
\newblock On learning graphs with edge-detecting queries.
\newblock In {\em {ALT}}, volume~98 of {\em Proceedings of Machine Learning
  Research}, pages 3--30. {PMLR}, 2019.

\bibitem[AB19b]{AssadiB19}
Sepehr Assadi and Aaron Bernstein.
\newblock Towards a unified theory of sparsification for matching problems.
\newblock In {\em {SOSA}}, volume~69 of {\em OASIcs}, pages 11:1--11:20.
  Schloss Dagstuhl - Leibniz-Zentrum f{\"{u}}r Informatik, 2019.

\bibitem[AB21]{AssadiB21}
Sepehr Assadi and Soheil Behnezhad.
\newblock On the robust communication complexity of bipartite matching.
\newblock In {\em {APPROX-RANDOM}}, volume 207 of {\em LIPIcs}, pages
  48:1--48:17. Schloss Dagstuhl - Leibniz-Zentrum f{\"{u}}r Informatik, 2021.

\bibitem[ABK{\etalchar{+}}04]{AlonBKRS04}
Noga Alon, Richard Beigel, Simon Kasif, Steven Rudich, and Benny Sudakov.
\newblock Learning a hidden matching.
\newblock {\em {SIAM} J. Comput.}, 33(2):487--501, 2004.

\bibitem[ACK19]{AssadiCK19}
Sepehr Assadi, Yu~Chen, and Sanjeev Khanna.
\newblock Polynomial pass lower bounds for graph streaming algorithms.
\newblock In {\em {STOC}}, pages 265--276. {ACM}, 2019.

\bibitem[ACK21]{AssadiCK21}
Sepehr Assadi, Deeparnab Chakrabarty, and Sanjeev Khanna.
\newblock Graph connectivity and single element recovery via linear and {OR}
  queries.
\newblock In {\em {ESA}}, volume 204 of {\em LIPIcs}, pages 7:1--7:19. Schloss
  Dagstuhl - Leibniz-Zentrum f{\"{u}}r Informatik, 2021.

\bibitem[AD21]{AssadiD21}
Sepehr Assadi and Aditi Dudeja.
\newblock A simple semi-streaming algorithm for global minimum cuts.
\newblock In {\em {SOSA}}, pages 172--180. {SIAM}, 2021.

\bibitem[AJJ{\etalchar{+}}22]{AssadiJJST22}
Sepehr Assadi, Arun Jambulapati, Yujia Jin, Aaron Sidford, and Kevin Tian.
\newblock Semi-streaming bipartite matching in fewer passes and optimal space.
\newblock In {\em {SODA}}, pages 627--669. {SIAM}, 2022.

\bibitem[AK20]{AhmadiK20}
Mohamad Ahmadi and Fabian Kuhn.
\newblock Distributed maximum matching verification in {CONGEST}.
\newblock In {\em {DISC}}, volume 179 of {\em LIPIcs}, pages 37:1--37:18.
  Schloss Dagstuhl - Leibniz-Zentrum f{\"{u}}r Informatik, 2020.

\bibitem[AKL17]{DBLP:conf/soda/AssadiKL17}
Sepehr Assadi, Sanjeev Khanna, and Yang Li.
\newblock On estimating maximum matching size in graph streams.
\newblock In {\em {SODA}}, pages 1723--1742. {SIAM}, 2017.

\bibitem[AKLY16]{AssadiKLY16}
Sepehr Assadi, Sanjeev Khanna, Yang Li, and Grigory Yaroslavtsev.
\newblock Maximum matchings in dynamic graph streams and the simultaneous
  communication model.
\newblock In {\em {SODA}}, pages 1345--1364. {SIAM}, 2016.

\bibitem[AKO18]{AhmadiKO18}
Mohamad Ahmadi, Fabian Kuhn, and Rotem Oshman.
\newblock Distributed approximate maximum matching in the {CONGEST} model.
\newblock In {\em {DISC}}, volume 121 of {\em LIPIcs}, pages 6:1--6:17. Schloss
  Dagstuhl - Leibniz-Zentrum f{\"{u}}r Informatik, 2018.

\bibitem[AL21]{AuzaL21}
Arinta Auza and Troy Lee.
\newblock On the query complexity of connectivity with global queries.
\newblock {\em CoRR}, abs/2109.02115, 2021.

\bibitem[ALT21]{AssadiLT21}
Sepehr Assadi, S.~Cliff Liu, and Robert~E. Tarjan.
\newblock An auction algorithm for bipartite matching in streaming and
  massively parallel computation models.
\newblock In {\em {SOSA}}, pages 165--171. {SIAM}, 2021.

\bibitem[Amb02]{Ambainis02}
Andris Ambainis.
\newblock Quantum lower bounds by quantum arguments.
\newblock {\em J. Comput. Syst. Sci.}, 64(4):750--767, 2002.

\bibitem[AMV20]{AxiotisMV20}
Kyriakos Axiotis, Aleksander Madry, and Adrian Vladu.
\newblock Circulation control for faster minimum cost flow in unit-capacity
  graphs.
\newblock In {\em {FOCS}}, pages 93--104. {IEEE}, 2020.

\bibitem[AR20a]{AgarwalR20}
Udit Agarwal and Vijaya Ramachandran.
\newblock Faster deterministic all pairs shortest paths in congest model.
\newblock In {\em {SPAA}}, pages 11--21. {ACM}, 2020.

\bibitem[AR20b]{AssadiR20}
Sepehr Assadi and Ran Raz.
\newblock Near-quadratic lower bounds for two-pass graph streaming algorithms.
\newblock In {\em {FOCS}}, pages 342--353. {IEEE}, 2020.

\bibitem[AV20]{AnariV20}
Nima Anari and Vijay~V. Vazirani.
\newblock Matching is as easy as the decision problem, in the {NC} model.
\newblock In {\em {ITCS}}, volume 151 of {\em LIPIcs}, pages 54:1--54:25.
  Schloss Dagstuhl - Leibniz-Zentrum f{\"{u}}r Informatik, 2020.

\bibitem[BdW02]{BuhrmanW02}
Harry Buhrman and Ronald de~Wolf.
\newblock Complexity measures and decision tree complexity: a survey.
\newblock {\em Theor. Comput. Sci.}, 288(1):21--43, 2002.

\bibitem[Ben22a]{BeniaminiUBPM}
Gal Beniamini.
\newblock Algebraic representations of unique bipartite perfect matching.
\newblock {\em CoRR}, abs/2203.01071, 2022.

\bibitem[Ben22b]{BeniaminiDegree}
Gal Beniamini.
\newblock The approximate degree of bipartite perfect matching.
\newblock {\em CoRR}, abs/2004:14318, 2022.

\bibitem[Ber09]{Bertsekas09}
Dimitri~P. Bertsekas.
\newblock Auction algorithms.
\newblock In {\em Encyclopedia of Optimization}, pages 128--132. Springer,
  2009.

\bibitem[BFS86]{BabaiFS86}
L{\'{a}}szl{\'{o}} Babai, Peter Frankl, and Janos Simon.
\newblock Complexity classes in communication complexity theory (preliminary
  version).
\newblock In {\em {FOCS}}, pages 337--347. {IEEE} Computer Society, 1986.

\bibitem[BHR{\etalchar{+}}18]{BeameHRRS18}
Paul Beame, Sariel Har{-}Peled, Sivaramakrishnan~Natarajan Ramamoorthy, Cyrus
  Rashtchian, and Makrand Sinha.
\newblock Edge estimation with independent set oracles.
\newblock In {\em {ITCS}}, volume~94 of {\em LIPIcs}, pages 38:1--38:21.
  Schloss Dagstuhl - Leibniz-Zentrum f{\"{u}}r Informatik, 2018.

\bibitem[BHR19]{BernsteinHR19}
Aaron Bernstein, Jacob Holm, and Eva Rotenberg.
\newblock Online bipartite matching with amortized \emph{O}(log
  \({}^{\mbox{2}}\) \emph{n}) replacements.
\newblock {\em J. {ACM}}, 66(5):37:1--37:23, 2019.

\bibitem[BLL{\etalchar{+}}21]{BrandLLSS0W21-maxflow}
Jan van~den Brand, Yin~Tat Lee, Yang~P. Liu, Thatchaphol Saranurak, Aaron
  Sidford, Zhao Song, and Di~Wang.
\newblock Minimum cost flows, {MDPs}, and $\ell_1$-regression in nearly linear
  time for dense instances.
\newblock In {\em {STOC}}, pages 859--869. {ACM}, 2021.

\bibitem[BLN{\etalchar{+}}20]{BrandLNPSSSW20}
Jan van~den Brand, Yin~Tat Lee, Danupon Nanongkai, Richard Peng, Thatchaphol
  Saranurak, Aaron Sidford, Zhao Song, and Di~Wang.
\newblock Bipartite matching in nearly-linear time on moderately dense graphs.
\newblock In {\em {FOCS}}, pages 919--930. {IEEE}, 2020.

\bibitem[BLSS20]{BrandLSS20}
Jan van~den Brand, Yin~Tat Lee, Aaron Sidford, and Zhao Song.
\newblock Solving tall dense linear programs in nearly linear time.
\newblock In {\em {STOC}}, pages 775--788. {ACM}, 2020.

\bibitem[BN19]{BernsteinN19}
Aaron Bernstein and Danupon Nanongkai.
\newblock Distributed exact weighted all-pairs shortest paths in near-linear
  time.
\newblock In {\em {STOC}}, pages 334--342. {ACM}, 2019.

\bibitem[BN21]{BeniaminiN21}
Gal Beniamini and Noam Nisan.
\newblock Bipartite perfect matching as a real polynomial.
\newblock In {\em {STOC}}, pages 1118--1131. {ACM}, 2021.

\bibitem[Bra20]{b20}
Jan van~den Brand.
\newblock A deterministic linear program solver in current matrix
  multiplication time.
\newblock In {\em SODA}, pages 259--278. {SIAM}, 2020.

\bibitem[CGL{\etalchar{+}}20]{ChuzhoyGLNPS20}
Julia Chuzhoy, Yu~Gao, Jason Li, Danupon Nanongkai, Richard Peng, and
  Thatchaphol Saranurak.
\newblock A deterministic algorithm for balanced cut with applications to
  dynamic connectivity, flows, and beyond.
\newblock In {\em {FOCS}}, pages 1158--1167. {IEEE}, 2020.

\bibitem[CK07]{ChakrabartiK07}
Amit Chakrabarti and Subhash Khot.
\newblock Improved lower bounds on the randomized complexity of graph
  properties.
\newblock {\em Random Struct. Algorithms}, 30(3):427--440, 2007.

\bibitem[CKL{\etalchar{+}}22]{ChenKLPGS22}
Li~Chen, Rasmus Kyng, Yang~P Liu, Richard Peng, Maximilian Probst~Gutenberg,
  and Sushant Sachdeva.
\newblock Maximum flow and minimum-cost flow in almost-linear time.
\newblock {\em CoRR}, abs/2203.00671, 2022.

\bibitem[CKLM18]{ChattopadhyayKL18}
Arkadev Chattopadhyay, Michal Kouck{\'{y}}, Bruno Loff, and Sagnik
  Mukhopadhyay.
\newblock Simulation beats richness: new data-structure lower bounds.
\newblock In {\em {STOC}}, pages 1013--1020. {ACM}, 2018.

\bibitem[CKP{\etalchar{+}}21]{ChenKPS0Y21}
Lijie Chen, Gillat Kol, Dmitry Paramonov, Raghuvansh~R. Saxena, Zhao Song, and
  Huacheng Yu.
\newblock Almost optimal super-constant-pass streaming lower bounds for
  reachability.
\newblock In {\em {STOC}}, pages 570--583. {ACM}, 2021.

\bibitem[CLS19]{CohenLS19}
Michael~B. Cohen, Yin~Tat Lee, and Zhao Song.
\newblock Solving linear programs in the current matrix multiplication time.
\newblock In {\em {STOC}}, pages 938--942. {ACM}, 2019.

\bibitem[CM20]{ChechikM20}
Shiri Chechik and Doron Mukhtar.
\newblock Single-source shortest paths in the {CONGEST} model with improved
  bound.
\newblock In {\em {PODC} '20: {ACM} Symposium on Principles of Distributed
  Computing, Virtual Event, Italy, August 3-7, 2020}, pages 464--473. {ACM},
  2020.

\bibitem[CMSV17]{cmsv17}
Michael~B Cohen, Aleksander Madry, Piotr Sankowski, and Adrian Vladu.
\newblock Negative-weight shortest paths and unit capacity minimum cost flow in
  ${O}(m^{10/7} \log {W})$ time.
\newblock In {\em {SODA}}, pages 752--771. SIAM, 2017.

\bibitem[CQ21]{ChekuriQ21a}
Chandra Chekuri and Kent Quanrud.
\newblock Isolating cuts, (bi-)submodularity, and faster algorithms for
  connectivity.
\newblock In {\em {ICALP}}, volume 198 of {\em LIPIcs}, pages 50:1--50:20.
  Schloss Dagstuhl - Leibniz-Zentrum f{\"{u}}r Informatik, 2021.

\bibitem[DEMN21]{DoryEMN21}
Michal Dory, Yuval Efron, Sagnik Mukhopadhyay, and Danupon Nanongkai.
\newblock Distributed weighted min-cut in nearly-optimal time.
\newblock In {\em {STOC}}, pages 1144--1153. {ACM}, 2021.

\bibitem[DHHM06]{DurrHHM06}
Christoph D{\"{u}}rr, Mark Heiligman, Peter H{\o}yer, and Mehdi Mhalla.
\newblock Quantum query complexity of some graph problems.
\newblock {\em {SIAM} J. Comput.}, 35(6):1310--1328, 2006.

\bibitem[DHK{\etalchar{+}}12]{SarmaHKKNPPW12}
Atish {Das Sarma}, Stephan Holzer, Liah Kor, Amos Korman, Danupon Nanongkai,
  Gopal Pandurangan, David Peleg, and Roger Wattenhofer.
\newblock Distributed verification and hardness of distributed approximation.
\newblock {\em {SIAM} J. Comput.}, 41(5):1235--1265, 2012.

\bibitem[DNO19]{DobzinskiNO19}
Shahar Dobzinski, Noam Nisan, and Sigal Oren.
\newblock Economic efficiency requires interaction.
\newblock {\em Games Econ. Behav.}, 118:589--608, 2019.

\bibitem[DP89]{DurisP89}
Pavol Duris and Pavel Pudl{\'{a}}k.
\newblock On the communication complexity of planarity.
\newblock In {\em {FCT}}, volume 380 of {\em Lecture Notes in Computer
  Science}, pages 145--147. Springer, 1989.

\bibitem[DS08]{DaitchS08}
Samuel~I. Daitch and Daniel~A. Spielman.
\newblock Faster approximate lossy generalized flow via interior point
  algorithms.
\newblock In {\em {STOC}}, pages 451--460. {ACM}, 2008.

\bibitem[Elk20]{Elkin20a}
Michael Elkin.
\newblock Distributed exact shortest paths in sublinear time.
\newblock {\em J. {ACM}}, 67(3):15:1--15:36, 2020.

\bibitem[FF56]{ford1956maximal}
Lester~Randolph Ford and Delbert~R Fulkerson.
\newblock Maximal flow through a network.
\newblock {\em Canadian journal of Mathematics}, 8:399--404, 1956.

\bibitem[FGL{\etalchar{+}}21]{ForsterGLPSY21}
Sebastian Forster, Gramoz Goranci, Yang~P. Liu, Richard Peng, Xiaorui Sun, and
  Mingquan Ye.
\newblock Minor sparsifiers and the distributed laplacian paradigm.
\newblock In {\em {FOCS}}, pages 989--999. {IEEE}, 2021.

\bibitem[FGT21]{FennerGT21}
Stephen~A. Fenner, Rohit Gurjar, and Thomas Thierauf.
\newblock Bipartite perfect matching is in quasi-nc.
\newblock {\em {SIAM} J. Comput.}, 50(3), 2021.

\bibitem[FKM{\etalchar{+}}05]{FeigenbaumKMSZ05}
Joan Feigenbaum, Sampath Kannan, Andrew McGregor, Siddharth Suri, and Jian
  Zhang.
\newblock On graph problems in a semi-streaming model.
\newblock {\em Theor. Comput. Sci.}, 348(2-3):207--216, 2005.

\bibitem[FN18]{ForsterN18}
Sebastian Forster and Danupon Nanongkai.
\newblock A faster distributed single-source shortest paths algorithm.
\newblock In {\em {FOCS}}, pages 686--697. {IEEE} Computer Society, 2018.

\bibitem[Gal16]{Gall16}
Fran{\c{c}}ois~Le Gall.
\newblock Further algebraic algorithms in the congested clique model and
  applications to graph-theoretic problems.
\newblock In {\em {DISC}}, volume 9888 of {\em Lecture Notes in Computer
  Science}, pages 57--70. Springer, 2016.

\bibitem[GG17]{GoldwasserG17}
Shafi Goldwasser and Ofer Grossman.
\newblock Bipartite perfect matching in pseudo-deterministic {NC}.
\newblock In {\em {ICALP}}, volume~80 of {\em LIPIcs}, pages 87:1--87:13.
  Schloss Dagstuhl - Leibniz-Zentrum f{\"{u}}r Informatik, 2017.

\bibitem[GHS83]{GallagerHS83}
Robert~G. Gallager, Pierre~A. Humblet, and Philip~M. Spira.
\newblock A distributed algorithm for minimum-weight spanning trees.
\newblock {\em {ACM} Trans. Program. Lang. Syst.}, 5(1):66--77, 1983.

\bibitem[GKK12]{GoelKK12}
Ashish Goel, Michael Kapralov, and Sanjeev Khanna.
\newblock On the communication and streaming complexity of maximum bipartite
  matching.
\newblock In {\em {SODA}}, pages 468--485. {SIAM}, 2012.

\bibitem[GKK{\etalchar{+}}18]{GhaffariKKLP18}
Mohsen Ghaffari, Andreas Karrenbauer, Fabian Kuhn, Christoph Lenzen, and Boaz
  Patt{-}Shamir.
\newblock Near-optimal distributed maximum flow.
\newblock {\em {SIAM} J. Comput.}, 47(6):2078--2117, 2018.

\bibitem[GKT01]{GabowKT01}
Harold~N. Gabow, Haim Kaplan, and Robert~Endre Tarjan.
\newblock Unique maximum matching algorithms.
\newblock {\em J. Algorithms}, 40(2):159--183, 2001.

\bibitem[GL18]{GhaffariL18}
Mohsen Ghaffari and Jason Li.
\newblock Improved distributed algorithms for exact shortest paths.
\newblock In {\em {STOC}}, pages 431--444. {ACM}, 2018.

\bibitem[GNT20]{Ghaffari0T20}
Mohsen Ghaffari, Krzysztof Nowicki, and Mikkel Thorup.
\newblock Faster algorithms for edge connectivity via random 2-out
  contractions.
\newblock In {\em {SODA}}, pages 1260--1279. {SIAM}, 2020.

\bibitem[GO16]{GuruswamiO16}
Venkatesan Guruswami and Krzysztof Onak.
\newblock Superlinear lower bounds for multipass graph processing.
\newblock {\em Algorithmica}, 76(3):654--683, 2016.

\bibitem[Gro96]{Grover96}
Lov~K. Grover.
\newblock A fast quantum mechanical algorithm for database search.
\newblock In {\em {STOC}}, pages 212--219. {ACM}, 1996.

\bibitem[Gr{\"u}60]{Grnbaum1960}
Branko Gr{\"u}nbaum.
\newblock Partitions of mass-distributions and of convex bodies by hyperplanes.
\newblock {\em Pacific Journal of Mathematics}, 10:1257--1261, 1960.

\bibitem[Haj91]{Hajnal91}
P{\'{e}}ter Hajnal.
\newblock An omega(n\({}^{\mbox{4/3}}\)) lower bound on the randomized
  complexity of graph properties.
\newblock {\em Comb.}, 11(2):131--143, 1991.

\bibitem[HHL18]{HatamiHL18}
Hamed Hatami, Kaave Hosseini, and Shachar Lovett.
\newblock Structure of protocols for {XOR} functions.
\newblock {\em {SIAM} J. Comput.}, 47(1):208--217, 2018.

\bibitem[HK73]{HopcroftK73}
John~E. Hopcroft and Richard~M. Karp.
\newblock An n\({}^{\mbox{5/2}}\) algorithm for maximum matchings in bipartite
  graphs.
\newblock {\em {SIAM} J. Comput.}, 2(4):225--231, 1973.

\bibitem[HKN21]{HenzingerKN21}
Monika Henzinger, Sebastian Krinninger, and Danupon Nanongkai.
\newblock A deterministic almost-tight distributed algorithm for approximating
  single-source shortest paths.
\newblock {\em {SIAM} J. Comput.}, 50(3), 2021.

\bibitem[HMT88]{HajnalMT88}
Andr{\'{a}}s Hajnal, Wolfgang Maass, and Gy{\"{o}}rgy Tur{\'{a}}n.
\newblock On the communication complexity of graph properties.
\newblock In {\em {STOC}}, pages 186--191. {ACM}, 1988.

\bibitem[HMT06]{HoangMT06}
Thanh~Minh Hoang, Meena Mahajan, and Thomas Thierauf.
\newblock On the bipartite unique perfect matching problem.
\newblock In {\em {ICALP}}, volume 4051 of {\em Lecture Notes in Computer
  Science}, pages 453--464. Springer, 2006.

\bibitem[HRVZ15]{HuangRVZ15}
Zengfeng Huang, Bozidar Radunovic, Milan Vojnovic, and Qin Zhang.
\newblock Communication complexity of approximate matching in distributed
  graphs.
\newblock In {\em {STACS}}, volume~30 of {\em LIPIcs}, pages 460--473. Schloss
  Dagstuhl - Leibniz-Zentrum f{\"{u}}r Informatik, 2015.

\bibitem[HRVZ20]{HuangRVZ20}
Zengfeng Huang, Bozidar Radunovic, Milan Vojnovic, and Qin Zhang.
\newblock Communication complexity of approximate maximum matching in the
  message-passing model.
\newblock {\em Distributed Comput.}, 33(6):515--531, 2020.

\bibitem[HWZ21]{HaeuplerWZ21}
Bernhard Haeupler, David Wajc, and Goran Zuzic.
\newblock Universally-optimal distributed algorithms for known topologies.
\newblock In {\em {STOC}}, pages 1166--1179. {ACM}, 2021.

\bibitem[II86]{IsraelI86}
Amos Israeli and Alon Itai.
\newblock A fast and simple randomized parallel algorithm for maximal matching.
\newblock {\em Inf. Process. Lett.}, 22(2):77--80, 1986.

\bibitem[IKL{\etalchar{+}}12]{IvanyosKLSW12}
G{\'{a}}bor Ivanyos, Hartmut Klauck, Troy Lee, Miklos Santha, and Ronald
  de~Wolf.
\newblock New bounds on the classical and quantum communication complexity of
  some graph properties.
\newblock In {\em {FSTTCS}}, volume~18 of {\em LIPIcs}, pages 148--159. Schloss
  Dagstuhl - Leibniz-Zentrum f{\"{u}}r Informatik, 2012.

\bibitem[JST20]{JinST20}
Yujia Jin, Aaron Sidford, and Kevin Tian.
\newblock Semi-streaming bipartite matching in fewer passes and less space.
\newblock {\em CoRR}, abs/2011.03495, 2020.

\bibitem[Kap21]{Kapralov21}
Michael Kapralov.
\newblock Space lower bounds for approximating maximum matching in the edge
  arrival model.
\newblock In {\em {SODA}}, pages 1874--1893. {SIAM}, 2021.

\bibitem[KKS14]{KapralovKS14}
Michael Kapralov, Sanjeev Khanna, and Madhu Sudan.
\newblock Approximating matching size from random streams.
\newblock In {\em {SODA}}, pages 734--751. {SIAM}, 2014.

\bibitem[KM93]{KushilevitzM93}
Eyal Kushilevitz and Yishay Mansour.
\newblock Learning decision trees using the fourier spectrum.
\newblock {\em {SIAM} J. Comput.}, 22(6):1331--1348, 1993.

\bibitem[KMNT20]{KapralovMNT20}
Michael Kapralov, Slobodan Mitrovic, Ashkan Norouzi{-}Fard, and Jakab Tardos.
\newblock Space efficient approximation to maximum matching size from uniform
  edge samples.
\newblock In {\em {SODA}}, pages 1753--1772. {SIAM}, 2020.

\bibitem[KMT21]{KapralovMT21}
Michael Kapralov, Gilbert Maystre, and Jakab Tardos.
\newblock Communication efficient coresets for maximum matching.
\newblock In {\em 4th Symposium on Simplicity in Algorithms, {SOSA} 2021,
  Virtual Conference, January 11-12, 2021}, pages 156--164. {SIAM}, 2021.

\bibitem[KN97]{KNBookCC}
Eyal Kushilevitz and Noam Nisan.
\newblock {\em Communication complexity}.
\newblock Cambridge University Press, 1997.

\bibitem[KP98]{KuttenP98}
Shay Kutten and David Peleg.
\newblock Fast distributed construction of small \emph{k}-dominating sets and
  applications.
\newblock {\em J. Algorithms}, 28(1):40--66, 1998.

\bibitem[KSS84]{KahnSS84}
Jeff Kahn, Michael~E. Saks, and Dean Sturtevant.
\newblock A topological approach to evasiveness.
\newblock {\em Comb.}, 4(4):297--306, 1984.

\bibitem[KUW85]{KarpUW85}
Richard~M. Karp, Eli Upfal, and Avi Wigderson.
\newblock Constructing a perfect matching is in random {NC}.
\newblock In {\em {STOC}}, pages 22--32. {ACM}, 1985.

\bibitem[KVKV11]{korte2011combinatorial}
Bernhard~H Korte, Jens Vygen, B~Korte, and J~Vygen.
\newblock {\em Combinatorial optimization}, volume~1.
\newblock Springer, 2011.

\bibitem[KVV85]{KozenVV85}
Dexter Kozen, Umesh~V. Vazirani, and Vijay~V. Vazirani.
\newblock {NC} algorithms for comparability graphs, interval gaphs, and testing
  for unique perfect matching.
\newblock In {\em {FSTTCS}}, volume 206 of {\em Lecture Notes in Computer
  Science}, pages 496--503. Springer, 1985.

\bibitem[KVV90]{KarpVV90}
Richard~M. Karp, Umesh~V. Vazirani, and Vijay~V. Vazirani.
\newblock An optimal algorithm for on-line bipartite matching.
\newblock In {\em {STOC}}, pages 352--358. {ACM}, 1990.

\bibitem[Lev65]{Levin65}
A.~Y. Levin.
\newblock On an algorithm for the minimization of convex functions over convex
  functions.
\newblock {\em Soviet Mathematics Doklady}, 160:1244--1247, 1965.

\bibitem[LL15]{LinL15}
Cedric~Yen{-}Yu Lin and Han{-}Hsuan Lin.
\newblock Upper bounds on quantum query complexity inspired by the
  elitzur-vaidman bomb tester.
\newblock In {\em Computational Complexity Conference}, volume~33 of {\em
  LIPIcs}, pages 537--566. Schloss Dagstuhl - Leibniz-Zentrum f{\"{u}}r
  Informatik, 2015.

\bibitem[LNP{\etalchar{+}}21]{LiNPSY21}
Jason Li, Danupon Nanongkai, Debmalya Panigrahi, Thatchaphol Saranurak, and
  Sorrachai Yingchareonthawornchai.
\newblock Vertex connectivity in poly-logarithmic max-flows.
\newblock In {\em {STOC}}, pages 317--329. {ACM}, 2021.

\bibitem[Lov79]{Lovasz79}
L{\'{a}}szl{\'{o}} Lov{\'{a}}sz.
\newblock On determinants, matchings, and random algorithms.
\newblock In {\em {FCT}}, pages 565--574. Akademie-Verlag, Berlin, 1979.

\bibitem[LP20]{LiP20}
Jason Li and Debmalya Panigrahi.
\newblock Deterministic min-cut in poly-logarithmic max-flows.
\newblock In {\em {FOCS}}, pages 85--92. {IEEE}, 2020.

\bibitem[LP21]{LiP21}
Jason Li and Debmalya Panigrahi.
\newblock Approximate gomory-hu tree is faster than \emph{n} - 1 max-flows.
\newblock In {\em {STOC}}, pages 1738--1748. {ACM}, 2021.

\bibitem[LS14]{LeeS14}
Yin~Tat Lee and Aaron Sidford.
\newblock Path finding methods for linear programming: Solving linear programs
  in {\~{o}}(sqrt(rank)) iterations and faster algorithms for maximum flow.
\newblock In {\em {FOCS}}, pages 424--433, 2014.

\bibitem[LS20]{LiuS20}
Yang~P. Liu and Aaron Sidford.
\newblock Faster energy maximization for faster maximum flow.
\newblock In {\em {STOC}}, pages 803--814. {ACM}, 2020.

\bibitem[LSW15]{LeeSW15}
Yin~Tat Lee, Aaron Sidford, and Sam~Chiu{-}wai Wong.
\newblock A faster cutting plane method and its implications for combinatorial
  and convex optimization.
\newblock In {\em {FOCS}}, pages 1049--1065. {IEEE} Computer Society, 2015.

\bibitem[LSZ20]{LiuSZ2020-stream}
S.~Cliff Liu, Zhao Song, and Hengjie Zhang.
\newblock Breaking the n-pass barrier: {A} streaming algorithm for maximum
  weight bipartite matching.
\newblock {\em CoRR}, abs/2009.06106, 2020.

\bibitem[Lub86]{Luby86}
Michael Luby.
\newblock A simple parallel algorithm for the maximal independent set problem.
\newblock {\em {SIAM} J. Comput.}, 15(4):1036--1053, 1986.

\bibitem[Mad13]{Madry13}
Aleksander Madry.
\newblock Navigating central path with electrical flows: From flows to
  matchings, and back.
\newblock In {\em {FOCS}}, pages 253--262. {IEEE} Computer Society, 2013.

\bibitem[Mad16]{Madry16}
Aleksander Madry.
\newblock Computing maximum flow with augmenting electrical flows.
\newblock In {\em {FOCS}}, pages 593--602. IEEE, 2016.

\bibitem[MN20]{MukhopadhyayN20}
Sagnik Mukhopadhyay and Danupon Nanongkai.
\newblock Weighted min-cut: sequential, cut-query, and streaming algorithms.
\newblock In {\em {STOC}}, pages 496--509. {ACM}, 2020.

\bibitem[MN21]{MukhopadhyayN21}
Sagnik Mukhopadhyay and Danupon Nanongkai.
\newblock A note on isolating cut lemma for submodular function minimization.
\newblock {\em CoRR}, abs/2103.15724, 2021.

\bibitem[MO09]{MontanaroO10}
Ashley Montanaro and Tobias Osborne.
\newblock On the communication complexity of {XOR} functions.
\newblock {\em CoRR}, abs/0909.3392, 2009.

\bibitem[MS04]{MuchaS04}
Marcin Mucha and Piotr Sankowski.
\newblock Maximum matchings via gaussian elimination.
\newblock In {\em {FOCS}}, pages 248--255. {IEEE} Computer Society, 2004.

\bibitem[MS20]{MandeS20}
Nikhil~S. Mande and Swagato Sanyal.
\newblock On parity decision trees for fourier-sparse boolean functions.
\newblock In {\em {FSTTCS}}, volume 182 of {\em LIPIcs}, pages 29:1--29:16.
  Schloss Dagstuhl - Leibniz-Zentrum f{\"{u}}r Informatik, 2020.

\bibitem[Nan14]{Nanongkai14}
Danupon Nanongkai.
\newblock Distributed approximation algorithms for weighted shortest paths.
\newblock In {\em {STOC}}, pages 565--573. {ACM}, 2014.

\bibitem[NC16]{NCQuantumBook}
Michael~A. Nielsen and Isaac~L. Chuang.
\newblock {\em Quantum Computation and Quantum Information (10th Anniversary
  edition)}.
\newblock Cambridge University Press, 2016.

\bibitem[New65]{newman1965location}
Donald~J Newman.
\newblock Location of the maximum on unimodal surfaces.
\newblock {\em Journal of the ACM (JACM)}, 12(3):395--398, 1965.

\bibitem[Nis21]{Nisan21}
Noam Nisan.
\newblock The demand query model for bipartite matching.
\newblock In {\em {SODA}}, pages 592--599. {SIAM}, 2021.

\bibitem[NS14]{NanongkaiS14}
Danupon Nanongkai and Hsin{-}Hao Su.
\newblock Almost-tight distributed minimum cut algorithms.
\newblock In {\em {DISC}}, volume 8784 of {\em Lecture Notes in Computer
  Science}, pages 439--453. Springer, 2014.

\bibitem[PS82]{PapadimitriouS82}
Christos~H. Papadimitriou and Michael Sipser.
\newblock Communication complexity.
\newblock In {\em {STOC}}, pages 196--200. {ACM}, 1982.

\bibitem[Raz92]{Razborov92}
Alexander~A. Razborov.
\newblock On the distributional complexity of disjointness.
\newblock {\em Theor. Comput. Sci.}, 106(2):385--390, 1992.

\bibitem[Ros73]{Rosenberg73}
Arnold~L. Rosenberg.
\newblock On the time required to recognize properties of graphs: a problem.
\newblock {\em {SIGACT} News}, 5(4):15--16, 1973.

\bibitem[Rot82]{Roth82}
Alvin~E. Roth.
\newblock The economics of matching: Stability and incentives.
\newblock {\em Math. Oper. Res.}, 7(4):617--628, 1982.

\bibitem[RSW22]{RoghaniSW22}
Mohammad Roghani, Amin Saberi, and David Wajc.
\newblock Beating the folklore algorithm for dynamic matching.
\newblock In {\em {ITCS}}, volume 215 of {\em LIPIcs}, pages 111:1--111:23.
  Schloss Dagstuhl - Leibniz-Zentrum f{\"{u}}r Informatik, 2022.

\bibitem[RV75]{RivestV75}
Ronald~L. Rivest and Jean Vuillemin.
\newblock A generalization and proof of the aanderaa-rosenberg conjecture.
\newblock In {\em {STOC}}, pages 6--11. {ACM}, 1975.

\bibitem[RV76]{RivestV76}
Ronald~L. Rivest and Jean Vuillemin.
\newblock On recognizing graph properties from adjacency matrices.
\newblock {\em Theor. Comput. Sci.}, 3(3):371--384, 1976.

\bibitem[RWZ20]{RashtchianWZ20}
Cyrus Rashtchian, David~P. Woodruff, and Hanlin Zhu.
\newblock Vector-matrix-vector queries for solving linear algebra, statistics,
  and graph problems.
\newblock In {\em {APPROX-RANDOM}}, volume 176 of {\em LIPIcs}, pages
  26:1--26:20. Schloss Dagstuhl - Leibniz-Zentrum f{\"{u}}r Informatik, 2020.

\bibitem[Ten02]{Tennenholtz02}
Moshe Tennenholtz.
\newblock Tractable combinatorial auctions and b-matching.
\newblock {\em Artif. Intell.}, 140(1/2):231--243, 2002.

\bibitem[Vai89]{Vaidya89}
Pravin~M. Vaidya.
\newblock A new algorithm for minimizing convex functions over convex sets
  (extended abstract).
\newblock In {\em {FOCS}}, pages 338--343. {IEEE} Computer Society, 1989.

\bibitem[VWW20]{VempalaWW20}
Santosh~S. Vempala, Ruosong Wang, and David~P. Woodruff.
\newblock The communication complexity of optimization.
\newblock In {\em {SODA}}, pages 1733--1752. {SIAM}, 2020.

\bibitem[Yao88]{Yao88}
Andrew~Chi{-}Chih Yao.
\newblock Monotone bipartite graph properties are evasive.
\newblock {\em {SIAM} J. Comput.}, 17(3):517--520, 1988.

\bibitem[Zha04]{Zhang04}
Shengyu Zhang.
\newblock On the power of ambainis's lower bounds.
\newblock In {\em {ICALP}}, volume 3142 of {\em Lecture Notes in Computer
  Science}, pages 1238--1250. Springer, 2004.

\end{thebibliography}
